\documentclass{article}
\usepackage[a4paper,bottom=3.5cm]{geometry}
\usepackage[sumlimits,namelimits,]{amsmath} 
\usepackage{amssymb}
\usepackage{mathabx}
\usepackage{ifthen}
\usepackage{xstring,xifthen}
\usepackage{stringstrings}
\usepackage[english]{babel}
\usepackage{spverbatim}
\usepackage{fancyhdr}
\usepackage{lastpage}
\usepackage{tikz-cd}
\usetikzlibrary{arrows}
\usepackage{amsthm}
\usepackage{thmtools}
\usepackage[colorinlistoftodos,prependcaption,textsize=tiny]{todonotes}
\usepackage[T1]{fontenc}
\usepackage[utf8]{inputenc} 
\usepackage{setspace}
\usepackage{enumitem}
\usepackage{stmaryrd}
\usepackage{tocloft}
\usepackage{comment}
\usepackage[hidelinks]{hyperref}
\usepackage[capitalize]{cleveref} 
\usepackage{siunitx}
\usepackage{indentfirst}
\usepackage{listings}
\usepackage{csquotes}
\usepackage[linesnumbered,rightnl,ruled]{algorithm2e}
\usepackage[noend]{algpseudocode}  
\usepackage{filecontents}
\usepackage{csquotes}
\usepackage{relsize}
\usepackage{xcolor}
\usepackage{autobreak}
\usepackage{makecell}
\usepackage{cancel}
\usepackage{bbold}
\usepackage{sagetex}
\SetInd{2.5em}{0.25em}
\SetAlgoVlined
\SetKwProg{Try}{try}{:}{}
\SetKwProg{Catch}{catch}{:}{}
\SetKwFor{IfCustom}{if}{:}{}
\SetKwFor{ElseCustom}{else}{:}{}
\SetKwFor{ForCustom}{for}{:}{}
\SetKwFor{WhileCustom}{while}{:}{}
\SetKwComment{Comment}{/* }{ */}
\vfuzz \hfuzz
\numberwithin{equation}{section} 

\newtheoremstyle{normal}   
{}                         
{}                         
{\normalfont}                            
{}                            
{\bfseries}                   
{}                    
{.5em}                        
{}                            

\newcommand{\hfilll}{%
  \hfill\hfill\hfill\hfill\hfill\hfill\hfill\hfill\hfill\hfill%
  \hfill\hfill\hfill\hfill\hfill\hfill\hfill\hfill\hfill\hfill%
  \hfill\hfill\hfill\hfill\hfill\hfill\hfill\hfill\hfill\hfill}
\theoremstyle{normal}
\fontencoding{T1}
\fontfamily{ptm}
\fontseries{m}
\fontshape{n}
\fontsize{12}{15}
\selectfont
\newtheorem{madefinition}{Definition}[section]
\newtheorem{maremarque}[madefinition]{Remark}

\newtheorem{montheoreme}[madefinition]{Theorem}
\newtheorem{maproposition}[madefinition]{Proposition}
\newtheorem{montheorem}[madefinition]{Theorem}

\newtheorem{monlemme}[madefinition]{Lemma}

\newtheorem{moncas}[madefinition]{Case}

\newtheorem{moncorollaire}[madefinition]{Corollary}
\newtheorem{monexample}[madefinition]{Example}

\newcommand{\Norm}[3]{\ensuremath\mathrm{Norm}_{#1/#2}(#3)}

\newcommand{\Trace}[3]{\ensuremath\mathrm{Trace}_{\ifx&#1&\else#1/#2\fi}(#3)}

\newcommand{\TraceRed}[3]{\ensuremath\mathrm{TraceRed}_{\ifx&#1&\else#1/#2\fi}(#3)}
\newcommand{\Adjoint}[1]{\ensuremath
\def\mystring{#1}
\StrLen{\mystring}[\mystringlen]
\def\mythresh{1}
\IfEq{\mystringlen}{\mythresh}{#1^{\star}}{(#1)^{\star}}
}

\newcommand{\End}[2]{\ensuremath\mathrm{End}_{#1}(#2)}

\newcommand{\BilPhi}[4]{\ensuremath(#3,\ifx&#2&#4\else#2(#4)\fi)_{#1}}
\newcommand{\BilPhiK}[4]{\ensuremath(#3,\ifx&#2&#4\else#2(#4)\fi)_{#1}}
\newcommand{\BilInd}[5]{\ensuremath\Trace{#1}{#2}{\BilPhi{\ifx&#3&#1\else#1,#3\fi }{}{#4}{#5}}}
\newcommand{\BilIndBracket}[4]{\ensuremath<#3,\ifx&#2&#4\else#2(#4)\fi>_{#1}}

\newcommand{\Ortho}[1]{\ensuremath#1^\perp}
\newcommand{\Frob}{\ensuremath\theta}
\newcommand{\Frobl}{\ensuremath\theta}
\newcommand{\Frobtl}{\ensuremath\theta}

\newcommand{\Mat}[2]{\mathrm{Mat}_{#1,#2}}

\newcommand{\GL}[2]{\mathrm{GL}_{#1}(#2)}
\newcommand{\Id}{\mathrm{id}}
\newcommand{\Eval}[1]{\mathrm{Eval}_{#1}}
\newcommand{\Trd}{T_{\mathrm{rd}}}

\algnewcommand{\algorithmicgoto}{\textit{go to}}%
\algnewcommand{\Goto}[1]{\algorithmicgoto~\ref{#1}}%

\newcommand{\stoptocwriting}{%
  \addtocontents{toc}{\protect\setcounter{tocdepth}{1}}}
\newcommand{\resumetocwriting}{%
  \addtocontents{toc}{\protect\setcounter{tocdepth}{\arabic{tocdepth}}}}

\newcommand{\qbin}[3]{\Bigg[\begin{array}{@{}c@{}}{#1}\\[-1ex] {#2}\end{array}\Bigg]_{#3}}
\newcommand{\Iso}{\text{\sc Iso}}
\newcommand{\IsoVect}{\text{iso}}
\newcommand{\IsoVects}{\IsoVect^\star}
\newcommand{\rgcd}{\text{rgcd}}
\newcommand{\llcm}{\text{llcm}}
\newcommand{\yield}{\text{\rm \bf yield}\xspace}
\newcommand{\nextc}{\text{\rm \bf next}\xspace}

\definecolor{babyblue}{rgb}{0.54, 0.81, 0.94}
\binoppenalty=\maxdimen
\relpenalty=\maxdimen

\begin{document}
\fussy
\sloppy
\title{Selfdual skew cyclic codes}
\author{Xavier Caruso \& Fabrice Drain}
\date{\today}
\maketitle
\begin{abstract}
    Given a finite extension $\mathbf{K/F}$ of degree $r$ of a finite field $\mathbf{F}$, we enumerate all selfdual skew cyclic codes in the Ore quotient ring $\mathbf{E}_{k}:=\mathbf{K}[X;\text{Frob}]/(X^{rk}-1)$ for any positive integer ${k}$ coprime to the characteristic $p$ (separable case).
    We also provide an enumeration algorithm when $k$ is a power of $p$ (purely inseparable case), at the cost of some redundancies.
    Our approach is based on an explicit bijection between skew cyclic codes, on the one hand, and certain families of $\mathbf{F}$-linear subspaces of some extensions of $\mathbf{K}$.
    Finally, we report on an implementation in SageMath.
\end{abstract}
\small
\setcounter{tocdepth}{2}
\tableofcontents
\stoptocwriting

\section{Introduction}
Among linear codes, cyclic codes enjoy a rich algebraic structure as they are defined as ideals of quotient polynomial rings. This structure endows them with good properties (encoding, decoding, duality, dimension, distance, length).
In the paper of Boucher, Geiselmann and Ulmer from 2006 \cite{17}, cyclic codes are generalized by considering left ideals in Ore polynomial rings rather than in polynomial rings, obtaining thus a much larger class of linear codes called skew cyclic codes.
In the present article, following their work, we study the selfdual property of these codes.

Let $\mathbf K/ \mathbf F$ be an extension of finite fields of degree $r$. Let $\Frob : \mathbf K \to \mathbf K$ be the Frobenius $x \mapsto x^q$ where $q$ denotes the cardinality of $\mathbf F$.
We consider the Ore polynomial ring $\mathbf{K}[X;\Frob]$, defined as the set of classical polynomials equipped with the standard addition and the twisted multiplication derived from the law $X\kappa= \Frob(\kappa)X$ where $\kappa$ is any element of $K$.
Skew cyclic codes are by definition left ideals of a quotient of the form $\mathbf{E}_{k}:=\mathbf{K}[X;\Frob]/(X^{rk} - 1)$.
We notice that cyclic codes correspond to the special case of skew cyclic codes when $r=1$.

The ambient space $\mathbf{E}_{k}$ is equipped with a bilinear form coming from the coordinatewise bilinear form on the vector space $\mathbf{K}^{rk}$, namely
$$\left(\sum_{i=0}^{kr-1} a_i X^i,\, \sum_{i=0}^{kr-1} b_i X^i\right) \mapsto \sum_{0\leq i < rk}  a_i b_i.$$
It thus makes sense to consider duality of skew cyclic codes. The topic was studied by Boucher among others~\cite{4,Bou2015,15,16}. 
Notably, an existence criterion is given in~\cite{Bou2015}, while an enumeration of selfdual skew cyclic codes for $r = 2$ and a prime field $\mathbf{F}$ appears in~\cite{15}.
In the subsequent article \cite{16}, the enumeration is extended to any nonnegative integer $r$, for ${k} = 1$ and a prime field $\mathbf{F}$.
In their conclusion, the authors suggest to further count and enumerate all selfdual skew cyclic codes for any values of the order $r$ and of the degree $k$ and for any finite base field $\mathbf{F}$.
In the present paper, we give a complete answer to this question when the characteristic $p$ of $\mathbf F$ is odd and ${k}$ is coprime to it (separable case).
We also study the case when $k$ is a $p$-th power (purely inseparable case) and obtain partial result in this case, our enumeration algorithm suffering from some redundancy.

As we will show in Subsection~\ref{section01}, $r$ has to be even for selfdual skew cyclic codes to exist. We thus set $r=2s$.
We first consider the separable case, \emph{i.e.} we assume that $k$ is coprime with $p$.
For the purpose of stating our main results, we write $\mathbf{F}[Y]/(Y^k - 1)$ as a product of field extensions of $\mathbf{F}$, namely $\mathbf{F}[Y]/(Y^k - 1) = \prod_{1 \leq l \leq n} \mathbf{F}_l$ where each $\mathbf{F}_l$ corresponds to an irreducible factor of $Y^k - 1$.
We let $y_l$ denote the image of $Y$ in $\mathbf{F}_l$ and we set $\mathbf{K}_l := \mathbf{K} \otimes_{\mathbf{F}} \mathbf{F}_l$. We also consider the function $\tau$ acting on the indices $l$ induced by the transformation $Y\mapsto \frac{1}{Y}$ on the irreducible factors of $Y^k - 1$.

Let $I$ be the subset of indexes $l \in \{1, \ldots, n\}$ which are fixed by $\tau$ and let $I_{\mathtt{eucl}}$ (resp. $I_{\mathtt{herm}}$) be the subset of $I$ consisting of indexes $l$ such that $y_l = \pm 1$ (resp. $y_l \neq \pm 1$).
Let also $J$ be the set consisting of all the pairs $\{l,\tau(l)\}$ when $l$ runs over the remaining indexes in $\{1, \ldots, n\} \setminus I$.
Finally, for each $l$, we denote by ${\mathcal{V}(\mathbf{K}_l/\mathbf{F}_l)}$ the set of $\mathbf{F}_l$-vector subspaces of $\mathbf{K}_l$.
When $l \in I_{\mathtt{eucl}}$ (resp. $l \in I_{\mathtt{herm}}$), we shall further equip $\mathbf{K}_l$ with a $\mathbf{F}_l$-bilinear Euclidean (resp. Hermitian) form; we let ${\mathcal{S}_{\mathtt{eucl}}(\mathbf{K}_l/\mathbf{F}_l)}$ (resp. ${\mathcal{S}_{\mathtt{herm}}(\mathbf{K}_l/\mathbf{F}_l)}$) denote the set of isotropic subspaces of $\mathbf{K}_l$ of dimension $s := r/2$.
Our main theorem is as follows.

\begin{montheorem}
\label{mainprop}
There exists an explicit bijection between the set of selfdual skew cyclic codes in $\mathbf{E}_{k}$ and the cartesian product of sets $W_{\mathtt{pal}}\times W_{\mathtt{nonpal}}$, where:
    \begin{align*}
    W_{\mathtt{pal}} = &\prod\limits_{l \in I_{\mathtt{eucl}}} {\mathcal{S}_{\mathtt{eucl}}}(\mathbf{K}_l/\mathbf{F}_l) \times  \prod\limits_{l \in I_{\mathtt{herm}}} {\mathcal{S}_{\mathtt{herm}}}(\mathbf{K}_l/\mathbf{F}_l)\\
    W_{\mathtt{nonpal}} = &\prod\limits_{\{l,\tau(l)\} \in J} {\mathcal{V}(\mathbf{K}_l/\mathbf{F}_l)}
    \end{align*}
\end{montheorem}
As a byproduct, we get the following counting of selfdual skew cyclic codes of $\mathbf{E}_{k}$.

\begin{montheorem}\label{maincount}
We assume that the characteristic of $\mathbf F$ is odd. Then there exist selfdual skew cyclic codes in $\mathbf{E}_k$ if and only if $k$ is odd, $s$ is odd and $q \equiv 3 \pmod 4$.
Moreover, when selfdual codes exist, their number is given by
$$\prod\limits_{l \in I_{\mathtt{eucl}}}\prod\limits_{i=0}^{s-1}\left(q_l^{i}+1\right)\times
  \prod\limits_{l \in I_{\mathtt{herm}}}\prod\limits_{i=0}^{s-1}\left(q_l^{i+1/2}+1\right)\times
  \prod\limits_{\{l,\tau(l)\} \in J} \sum\limits_{k=0}^{r}\frac{\left(q_l^{r}-1\right) \ldots\left(q_l^{r-k+1}-1\right)}{\left(q_l^{k}-1\right) \ldots\left(q_l-1\right)}$$ 
where $q_l$ denotes the cardinal of $\mathbf{F}_l$.
\end{montheorem}

The existence criterion stated in the previous theorem was already known~\cite{Bou2015}.
Our approach, however, is new and allows for deriving other informations, including the above counting formula and explicit methods for generating self-dual codes that we discuss now.
First of all, we describe algorithms, with polynomial complexity in $k$ and $r$, for generating randomly such a code, with uniform distribution.
We then move to the question of complete enumeration. Since selfdual skew cyclic codes are quite numerous (their number grows exponentially with respect to $r$), it sounds not that interesting to design an algorithm that outputs the complete list of such codes in one shot. Instead, we describe a routine that outputs a new code each time it is called with the guarantee that all codes will show up--and show up only once--at the end of the day. The cost of each indivual call to our algorithm is again polynomial in $k$ and $r$.

Our method looks robust in the sense that we are confident that it could be adapted to other situations, \emph{e.g.} even characteristic or negacyclic (or more generally, constacyclic) codes instead of cyclic codes.
However, addressing the inseparable case where $k$ is not coprime to $p$ using analogue methods seems more delicate (although probably doable).
In the present paper, we outline a different method for enumerating all purely inseparable selfdual skew cyclic codes, for which $k$ is a power of the characteristic $p$, by multiplying properly generators of some twisted separable selfdual skew cyclic codes as described and illustrated by hand of SageMath computations in Section~\ref{inseparableCase}.
This enumeration method could be used in combination with the enumeration method of the separable case to solve the general inseparable enumeration problem.
However, it is not optimal as it comes with redundancies.

\paragraph{Organization of the paper.}

In Section \ref{section0}, we define selfdual skew cyclic codes. Then, under the separability hypothesis that $k$ is coprime to $p$, we relate the skew algebra $\mathbf{E}_{k}$ to a product of matrix algebras, and we transport the bilinear structure of $\mathbf{E}_{k}$ onto matrices.
In Section \ref{section:algo}, we use this reinterpretation to count selfdual skew cyclic codes and to generate them efficiently.
In Section \ref{section4}, we report on an implementation of our algorithms and provide some numerical experiments.
The source code of the SageMath implementation is available at this location:

\begin{center}
\url{https://plmlab.math.cnrs.fr/caruso/selfdual-skew-cyclic-codes}\hfill
\end{center}

\noindent
In Section \ref{inseparableCase}, we sketch an enumeration algorithm for purely inseparable selfdual skew cyclic codes, in the case where $k$ is a power of $p$. Finally we provide computation results for the enumeration of purely inseparable skew cyclic codes.

\paragraph{Conventions and notations.}\label{notations}

Throughout this paper, we will use the following notation:
\begin{itemize}
    \item $\End{R}{V}$ denotes, for any ring $R$ and $R$-module $V$, the endomorphism ring of $R$-linear endomorphisms of $V$.
    \item $\Mat{R}{r \times r}$ denotes, for any ring $R$, the matrix ring of $r \times r$ square matrices with entries in $R$.
    \item $M^{\mathsf{tr}}$ denotes the transpose of the matrix $M$.
    \item $\Id$ denotes the identity morphism.
    \item $\GL n F$ denotes the general linear group of the vector space $F^n$ over the finite field $F$.
    \item $L^{\sigma}$ denotes the subfield of $L$ fixed by the automorphism $\sigma$.
    \item $\Ortho{V}$ denotes the orthogonal of the subvector space $V$ in some ambient space.
\end{itemize}

\noindent
If $F$ is a finite field, equipped with an involutive automorphism $\sigma$ (\emph{i.e.} an automorphism of order $2$),
we recall that a \textit{$\sigma$-sesquilinear form} $\mathcal{B}$ of a $F$-vector space $V$ is an additive map $\mathcal{B} : V \times V \to F$ such that
\begin{center}$\mathcal{B}(\lambda u, \mu v) = \lambda\cdot \sigma(\mu) \cdot \mathcal{B}(u, v)\qquad\forall u, v \in V,\quad\forall \lambda,\mu \in F$\end{center}  
In this paper, we will consider four different types of sesquilinear forms:
\begin{itemize}
\item (Euclidean case) $\sigma = \Id $ and $\mathcal B$ is symmetric, \emph{i.e.} $\mathcal{B}(u, v) = \mathcal{B}(v, u)$ for all $u,v \in V$,
\item (skew-Euclidean case) $\sigma = \Id $ and $\mathcal B$ is antisymmetric, \emph{i.e.} $\mathcal{B}(u, v) = -\mathcal{B}(v, u)$ for all $u,v \in V$,
\item (Hermitian case) $\sigma \neq \Id $ and $\mathcal B$ is symmetric,
\item (skew-Hermitian case) $\sigma \neq \Id $ and $\mathcal B$ is antisymmetric.
\end{itemize}

\noindent
We recall that, when $\mathcal B$ is nondegenerate, the ring $\End{F}{V}$ of $F$-linear endomorphisms of $V$ is equipped with an involutive anti-automorphism $f \mapsto \Adjoint{f}$ characterized by 
\begin{center}$\forall u,v \in V, \quad \mathcal{B}\big(u,\Adjoint{f}(v)\big) = \mathcal{B}\big(f(u),v\big)$\end{center}
\noindent
It is called the \emph{adjunction} relative to $\mathcal{B}$.
We recall that $\Adjoint{f+g} = \Adjoint{f} + \Adjoint{g}$ and $\Adjoint{f\circ g} = \Adjoint{g} \circ \Adjoint{f}$ for $f, g \in \End{F}{V}$.
Moreover, the adjoint of the scalar multiplication by an element $a \in F$ is the multiplication by $\sigma(a)$.
The adjoint construction allows finally to endow $\End{F}{V}$ itself with a sesquilinear pairing, defined by $\left< f, g \right> = \mathrm{Trace}(f \circ \Adjoint{g})$.


\paragraph{Acknowledgments.}

We warmly thank Delphine Boucher and the anonymous referee for pointing out several inconstancies in a former version of this article.

This research was supported by the grant ANR-21-CE39-0009-BARRACUDA.

\resumetocwriting
\section{From skew cyclic codes to finite geometry} \label{section0}
\subsection{Definition of skew cyclic codes}\label{section01}
Let $\mathbf{F}$ be a finite field of cardinality $q$ and characteristic $p$.
Let $\mathbf{K}$ be a finite extension of $\mathbf{F}$ of degree $r$.
Let $\mathbf{\Frob}: x \mapsto x^q$ be the Frobenius automorphism of $\mathbf{K}$. We build the quotient of the free $\mathbf{K}$-algebra $\mathbf{K}{\left<X\right>}$ by the noncommutative relation:
$\forall \kappa \in \mathbf{K}, \: X \kappa= \Frob(\kappa) X$. We then localize it at the powers of $X$.
This results in the Ore Laurent polynomial ring $\mathbf{K}[X^{\pm 1};\Frob]$.
As shown in \cite[Theorem 1.1.22]{10}, its center is $\mathbf{F}[X]\cap \mathbf{K}[X^{\pm r}] =\mathbf{F}[X^{\pm r}]$.
For any $f \in \mathbf{F}[X^{\pm r}]$, we can thus form the quotient $\mathbf{K}[X^{\pm 1};\Frob]/(f(X))$, which keeps a ring structure.
We will call \textit{skew quotient algebra} the algebra $\mathbf{K}[X^{\pm 1};\Frob]/(f(X))$.
\begin{maremarque}
    As a quotient ring of the left and right Euclidean domain of skew Laurent polynomials, $\mathbf{K}[X^{\pm 1};\Frob]$, any skew quotient algebra is a left and right principal ideal ring.
\end{maremarque}
We now move to the definition of selfdual skew cyclic codes.
For any nonnegative integer ${k}$, $X^{rk}-1$ is in the center of $\mathbf{K}[X^{\pm 1};\Frob]$.
We can thus form the quotient ring $\mathbf{E}_k := \mathbf{K}[X^{\pm 1};\Frob]/(X^{rk}-1)$.
Choosing for any element of $\mathbf{E}_{k}$ the unique lift in $\mathbf{K}[X;\Frob] \subset \mathbf{K}[X^{\pm 1};\Frob]$ of degree strictly less than $kr$ defines an isomorphism of $\mathbf{K}$-vector spaces
$\lambda:  \mathbf{E}_{k} \rightarrow \mathbf{K}^{rk}$.

Using the classical Hamming distance $d$ on the $\mathbf{K}$-vector space $\mathbf{K}^{rk}$, we define the \textit{Hamming distance} $D$ between two elements $f$ and $g$ of $\mathbf{E}_{k}$ by $D(f,g)=d(\lambda(f),\lambda(g))$.
\begin{madefinition}
Given $\alpha \in \mathbf{F}^*$, a \textit{skew $\alpha$-constacyclic codes} is a left ideal of $\mathbf{K}[X^{\pm 1};\Frob]/(X^{rk}-\alpha)$ endowed with the metric $D$.
\textit{Skew cyclic codes} (resp. \textit{skew negacyclic codes}) are skew $\alpha$-constacyclic codes for $\alpha = 1$ (resp. $\alpha = -1$).
\end{madefinition}

We are interested in the skew cyclic code duality for the coordinatewise bilinear form, defined on $\mathbf{K}^{rk}$~by
$$\big((x_i)_{0\leq i < rk},(y_i)_{0\leq i < rk}\big) \mapsto \sum_{0\leq i < rk}  x_i y_i$$
We note that this bilinear form is nondegenerate.

\begin{madefinition}
A skew cyclic code is said \textit{self-orthogonal} (resp. \textit{selfdual}) if $\lambda(I) \subset \Ortho{\lambda(I)}$
(resp. if $\lambda(I) = \Ortho{\lambda(I)}$).
\end{madefinition}

As we have $\dim(\lambda(I))+\dim(\Ortho{\lambda(I)})=r$, a necessary condition for selfdual skew cyclic codes to exist is that $r$ is even.

\subsection{The evaluation isomorphism $\mathcal{E}_l$}
\label{ssec:evalisom}

We now place ourselves in the separable case, where $k$ is coprime to $p$. 
It is then known that $\mathbf{E}_{k}$ is a semisimple algebra (see \cite[Proposition 20.7]{6}).
As $\mathbf{E}_{k}$ is finite-dimensional over $\mathbf{F}$, classical results imply that it is a cartesian product of matrix algebras over finite field extensions of $\mathbf{F}$.
Hereunder, we describe an explicit isomorphism realizing this decomposition.

We note $Y:=X^r$ and decompose $Y^k-1$ as a product of irreducible polynomials $P_l(Y)$ over $\mathbf{F}$.
We set ${\mathbf{F}_l}:= \mathbf{F}[Y]/P_l(Y)$, ${\mathbf{K}_l} := \mathbf{K}[Y]/P_l(Y)$ and let $y_{l}$ denote the image of $Y$ in $\mathbf{K}_l$.
We extend $\theta$ to an automorphism of $\mathbf{K}_l$ by letting it act trivially on $y_l$. We have a first decomposition
\begin{align}
  \mathbf{E}_{k} 
& \simeq \mathbf{K}[Y,X;\Frob]/(Y^k-1,X^{r}-Y) \nonumber \\
& \simeq \left(\frac{\mathbf{K}[Y,X;\Frob]}{\prod\limits_{1 \leq l \leq n} P_l(Y)}\right)/(X^{r}-Y) 
  = \prod\limits_{1 \leq l \leq n} \mathbf{K}_l[X^{\pm 1};\theta]/(X^{r}-y_l). \label{eq:decompositionEk}
\end{align}
We set $\tilde{\mathbf{E}}_{k}^{(l)} = \mathbf{K}_l[X^{\pm 1};\theta]/(X^{r}-y_l)$ and now study each $\tilde{\mathbf{E}}_{k}^{(l)}$ separately.
We observe that ${\mathbf{K}_l}$ is a finite étale extension of the finite field ${\mathbf{F}_l}$, \emph{i.e.} a finite product of finite extensions of ${\mathbf{F}_l}$.
As such, the norm map $\operatorname{Norm}_{\mathbf{K}_l/\mathbf{F}_l}$ is surjective; hence, there exists an element $x_{l}$ in ${\mathbf{K}_l}$ satisfying
$\Norm{{\mathbf{K}_l}}{\mathbf{F}_l}{x_l} = y_l$. The change of variables $X \mapsto x_l X$ defines an isomorphism
$$\Eval{{x_{l}}X} : 
 \tilde{\mathbf{E}}_{k}^{(l)} \stackrel\sim\longrightarrow
 \mathbf{E}_{k}^{(l)} := \mathbf{K}_l[X^{\pm 1};\theta]/(X^{r}-1).$$
On the other hand, we have an evaluation morphism $X \mapsto \Frobl$:
$$\begin{array}{rcl}
\Eval{\Frob} : \quad
  \mathbf{E}_{k}^{(l)} & \longrightarrow & \End{\mathbf{F}_l}{{\mathbf{K}_{l}}} \\
    P(X) & \mapsto & P(\Frobl)
\end{array}$$
Composing both maps, we obtain a third morphism
$\mathcal E_l :
 \tilde{\mathbf{E}}_{k}^{(l)} \to
 \End{\mathbf{F}_l}{{\mathbf{K}_{l}}}$.
Applying it to each term $\tilde{\mathbf{E}}_{k}^{(l)}$ of the decomposition~\eqref{eq:decompositionEk},
we finally end up with a map relating $\mathbf{E}_{k}$ to a product of matrix algebras.

\begin{maproposition}
The map
$$\big(\mathcal{E}_l\big)_{l \in \{1,\ldots,n\}}: 
  \mathbf{E}_{k} \longrightarrow \prod_{1\leq l \leq n} \End{\mathbf{F}_l}{{\mathbf{K}_{l}}}$$
is an isomorphism of ${\mathbf{F}}$-algebra. 
\end{maproposition}

\begin{proof}
(See also \cite[Theorem~1.3.12]{10}.)
    By Artin's lemma, the family $(\Frobl^i)_{0\leq i <r}$ is linearly independent. This proves the injectivity of $\mathcal E_l$.
    As the dimension (over $\mathbf F_l$) of its domain and the codomain are both $r^2$, surjectivity follows.
    The final evaluation map resulting from the composition of the chinese remainder isomorphism with the product of isomorphisms $\Eval{x_{l}\Frobl}$ is thus an isomorphism.
\end{proof}
\begin{maremarque}
    To compute the evaluation isomorphism $\mathcal{E}_l$, a fast computation of preimages by the norm is needed.
    One possible method consists in finding an irreducible factor of the skew polynomial $X^r - y_l$ in $\mathbf{K}_l[X;\Frobl]$. An algorithm for this task is described in \cite{18}.
\end{maremarque}
\begin{maremarque}
    By the Skolem-Noether theorem, the isomorphism $\mathcal{E}_l$ is uniquely defined up to conjugacy by an element of norm $1$, \emph{i.e.} up to another choice of $x_l$ as preimage of $y_l$ by the norm map.
\end{maremarque}
\subsection{Adjunctions on $\mathbf{E}_{k}$ and related spaces}
\label{ssec:adjunction}
In this subsection, we construct an alternative bilinear pairing on $\mathbf{E}_{k}$ and show that it induces the same orthogonals than the coordinatewise bilinear form considered previously.
Our variant is interesting because it will in turn induce a pairing on the simpler spaces $\mathbf{E}_{k}^{(l)}$.

We begin by defining an adjunction on $\mathbf{E}_{k}$.
We start from the following $\mathbf{F}$-linear automorphism on $\mathbf{K}[X^{\pm 1};\Frob]$
$$\begin{array}{rcl}
  \mathbf{K}[X^{\pm 1};\Frob] & \overset{*}{\longrightarrow} & \mathbf{K}[X^{\pm 1};\Frob] \\
  f = \sum_i f_i X^i & \mapsto & f^* = \sum_i X^{-i} f_i
\end{array}$$
It is an involution, \emph{i.e.} $f^{**} = f$ for all $f$.
One moreover checks that it is an anti-morphism, \emph{i.e.} it satisfies $(fg)^* = g^* f^*$
for all $f, g \in \mathbf{K}[X^{\pm 1};\Frob]$. Indeed, by linearity, it is enough to check the desired property when $f$ and $g$ are monomials, which is a direct computation.
We observe that the adjoint $(X^{rk}-1)^*$ is a multiple of $X^{rk}-1$ itself. The adjunction thus preserves the two-sided ideal generated by $X^{rk}-1$; therefore, it passes to the quotient to define an anti-automorphism of $\mathbf{E}_k$.
In a slight abuse of notation, we continue to write $f^*$ when $f \in \mathbf{E}_k$.
Since the adjunction is an anti-automorphism, we underline that it maps left ideals of $\mathbf{E}_{k}$ to right ideals.

We now define a nondegenerate bilinear form corresponding to this adjunction.
We introduce to this end the notion of \emph{reduced trace} of a skew polynomial: given $f = \sum_i f_i X^i \in \mathbf{K}[X^{\pm 1};\Frob]$, we set
$$\Trd(f) = \sum_i \Trace{\mathbf{K}}{\mathbf{F}}{f_{ir}}{\cdot}Y^i \in \mathbf F[Y].$$
This definition passes to the quotient and defines a map $\Trd : \mathbf{E}_{k} \to \mathbf{F}[Y]/(Y^k - 1)$. For $f, g \in \mathbf{E}_{k}$, we define
$$\langle f,g\rangle :=\Trd(fg^*) \in \mathbf{F}[Y]/(Y^k - 1).$$
It is readily seen that $f \mapsto f^*$ satisfies the adjunction formula, in the sense that $\langle f, gh\rangle = \langle fh^*,g\rangle$ for any $f, g,h \in \mathbf{E}_{k}$. 
Denoting by $\Ortho{I}$ the orthogonal of $I \subset \mathbf{E}_k$, we have the following compatibility property.

\begin{maproposition}
For any left ideal $I$ of $\mathbf{E}_{k}$, we have $\lambda(\Ortho{I})=\Ortho{\lambda(I)}$. 
\end{maproposition}
\begin{proof}
    Let $f \in I$, $g \in \Ortho{I}$ and set $u = \sum_{0 \leq i < kr} \lambda(f)_i \lambda(g)_{i}$.
    The condition $\Trd(fg^*) = 0$ translates to $\Trace{\mathbf{K}}{\mathbf{F}} u = 0$.
    More generally, for all $\kappa \in \mathbf{K}$, we have $\Trd(\kappa fg^*) = 0$ (since $\kappa f \in I$) and so $\Trace{\mathbf{K}}{\mathbf{F}}{\kappa u} = 0$.
    By nondegeneracy of $\mathrm{Trace}_{\mathbf{K}/\mathbf{F}}$, this implies that $u = 0$, \emph{i.e.} $\lambda(f)$ is orthogonal to $\lambda(g)$.
    The inclusion $\lambda(\Ortho{I}) \subset \Ortho{\lambda(I)}$ follows.

    Conversely, we consider $g \in \mathbf{E}_k$ such that $\lambda(g) \in \Ortho{\lambda(I)}$.
    Let also $f \in I$.
    The orthogonality condition between $\lambda(f)$ and $\lambda(g)$ implies the vanishing of the constant coefficient of $f g^*$.
    More generally, using that $X^{-ir} f \in I$, we deduce that the coefficient in $X^{ir}$ of $f g^*$ vanishes as well.
    Therefore $\Trd(f g^*) = 0$ and we conclude that $g \in \Ortho{I}$.
\end{proof}
In what follows, we prefer working with the pairing $\langle -,-\rangle$ because it corresponds to sesquilinear trace forms on the ${\mathbf{F}_l}$-algebras $\mathbf{E}_{k}^{(l)}$.
We now describe them.
\begin{madefinition}
    We say that a polynomial is \textit{palindromic} if the set of its roots in an algebraic closure of its base field does not contain zero and is stable under the inversion map $x \mapsto \frac{1}{x}$.
    Equivalently a polynomial $\sum_{i=0}^{n} p_i x^i$ (with $p_n \neq 0$) is \textit{palindromic} if it is collinear to its reciprocal polynomial $\sum_{i=0}^{n} p_{n-i} x^i$.
\end{madefinition}
We recall that we have the decomposition
$F[Y]/(Y^k-1) \simeq \prod\limits_{1 \leq l \leq n} F[Y]/P_l(Y) \simeq \prod\limits_{1 \leq l \leq n} \mathbf{F}_l$.
\begin{madefinition}\label{tau}
    We define the map $\tau : \{1,\ldots,n\} \to \{1,\ldots,n\}$ by the relation $P_{\tau(l)}(\frac{1}{y_{l}})=0$.
\end{madefinition}
As the polynomial $Y^k-1$ is palindromic and separable, the index $\tau(l)$ exists, and $\tau$ is obviously involutive.
Furthermore, we let $\sigma$ be the endomorphism of $F$-algebras of $\mathbf{F}[Y]/(Y^k-1)$ defined by $Y \mapsto \frac 1 Y$.
It is also involutive and it induces isomorphisms $\sigma_l : {\mathbf{F}_l} \to {\mathbf{F}_{\tau(l)}}$ such that $\sigma_l(y_{l}) = y_{\tau(l)}$.
The tensor product $\Id \otimes \sigma_l$ defines an involutive isomorphism $\mathbf{K}_l \to \mathbf{K}_{\tau(l)}$ extending $\sigma_l$.
For simplicity, we will keep the notation ${{\sigma}}_l$ for $\Id \otimes \sigma_l$.

The next proposition shows that the adjunction behaves nicely with respect to the decomposition $\mathbf{E}_{k} = \prod_{l=1}^n \tilde{\mathbf{E}}_{k}^{(l)}$ we have established in Equation~\eqref{eq:decompositionEk}.

\begin{maproposition}
The adjunction $f \mapsto f^*$ induces ``partial'' adjunctions
$\tilde{\mathbf{E}}_{k}^{(l)} \to \tilde{\mathbf{E}}_{k}^{(\tau(l))}$, which are explicitly given by the formula
\begin{equation}
\label{eq:partialadjunction}
   \sum_{i=0}^{\deg{P_l}-1} f_i X^i \mapsto \sum_{i=0}^{\deg{P_l}-1} X^{-i} {{\sigma}}_l(f_i), \quad \forall f_i \in \mathbf{K}_l.
\end{equation}
Moreover, the ``global'' adjunction can be recovered by taking the product of the partial ones.
\end{maproposition}

\begin{proof}
Let $Q_l$ be the element of $F[Y]/(Y^k-1) \subset \mathbf{E}_{k}$ corresponding to the factor $\mathbf{F}_l$, \emph{i.e.} the element defined by the congruences $Q_l \equiv 1 \pmod{P_l}$ and $Q_l \equiv 0 \pmod{P_{l'}}$ whenever $l' \neq l$.
As automorphisms respect congruences, we have $Q_l^* = \sigma(Q_l) = Q_{\tau(l)}$.
Besides $\tilde{\mathbf{E}}_{k}^{(l)} = Q_l \mathbf{E}_{k} = \mathbf{E}_{k} Q_l$. We thus have
$\tilde{\mathbf{E}}_{k}^{(l)}{}^* = (Q_l \mathbf{E}_{k})^* = \mathbf{E}_{k} Q_l^* = \mathbf{E}_{k} Q_{\tau(l)} = \tilde{\mathbf{E}}_{k}^{(\tau(l))}$.
The explicit formula~\eqref{eq:partialadjunction} is derived after noticing that $f_i^* = \sigma_l(f_i)$ for $f_i \in \mathbf K_l$.
Finally, the last statement of the proposition is clear.
\end{proof}

We now aim at describing how the adjunction is transformed by the evaluation isomorphisms $\mathcal{E}_l$.
For this, the first step is to understand its effect on $\mathbf{E}_{k}^{(l)}$ (without the tilde) which, we recall, is defined as $\mathbf{E}_{k}^{(l)} = \mathbf K_l[X;\theta]/(X^r-1)$.
The adjunction $f \mapsto f^*$ again passes to the quotient and determines a well-defined adjunction $\mathbf{E}_{k}^{(l)} \to \mathbf{E}_{k}^{(\tau(l))}$, that we continue to denote $f \mapsto f^*$.
Unfortunately, the latter is not exactly what we need; we are now going to fix this issue by defining a twisting version of it.
For this, we first define $z_l := x_l\cdot{{\sigma}}_{\tau(l)}(x_{\tau(l)}) \in \mathbf{K}_l$.

\begin{monlemme}
\label{lem:zetal}
There exists a family of invertible elements $\zeta_l \in \mathbf{K}_l$ such that
$\Frob(\zeta_l)= z_l \zeta_l$ and ${{\sigma}}_l(\zeta_l)=\zeta_{\tau(l)}$ for all $l$.
\end{monlemme}
\begin{proof}
    Since $\sigma_l \circ \sigma_{\tau(l)} = \Id$, we have ${{\sigma}}_{l}(x_l\cdot{{\sigma}}_{\tau(l)}(x_{\tau(l)}))=x_{\tau(l)}{{\sigma}}_{l}(x_l)$, which ensures that $z_l$ is invariant under ${{\sigma}}_{l}$.
    Furthermore, we observe that
    $\Norm{\mathbf{K}_l}{\mathbf{F}_l}{z_l} = y_l \cdot \sigma_{\tau(l)}(y_{\tau(l)}) = 1$.
    Hence, the Hilbert 90 Theorem guarantees the existence of an element $\zeta_l$ of $\mathbf{K}_l^*$ such that $\Frob(\zeta_l)= z_l \zeta_l$ and hence $z_l X = \zeta_l^{-1} X \zeta_l$.
    Moreover, as automorphisms of finite fields commute, ${{\sigma}}_l(\zeta_l)$ satisfies $\Frob({{\sigma}}_l(\zeta_l))=z_{l} {{\sigma}}_l(\zeta_l)$.
    Set ${\zeta_l'}:=\zeta_l + {{\sigma}}_{\tau(l)}(\zeta_{\tau(l)})$, so that we have ${{\sigma}}_l(\zeta'_l)=\zeta'_{\tau(l)}$.
    If $\zeta'_l$ is invertible, it satisfies $z_l X = {\zeta_l'}^{-1} X {\zeta_l'}$ as well.
    At the opposite side, if ${\zeta_l'}=0$, we have ${{\sigma}}_{\tau(l)}(\zeta_{\tau(l)})=-\zeta_l$. In this case, ${{\sigma}}_l$ is nontrivial and so $y_{\tau(l)} \neq \pm 1$.
    As $\frac{\zeta_l}{y_l}$ satisfies also $\Frob\big(\frac{\zeta_l}{y_l}\big)=z_l \frac{\zeta_l}{y_l}$, the element ${\zeta_l''} := \frac{\zeta_l}{y_l} + {{\sigma}}_{\tau(l)}\big(\frac{\zeta_{\tau(l)}}{y_{\tau(l)}}\big)$ does not vanish and it satisfies moreover $\zeta_l'' z_l X = X {\zeta_l''}$ and $\sigma_l(\zeta''_l)=\zeta''_{\tau(l)}$.
    In full generality, we write $\mathbf{K}_l$ as a product of fields and apply the previous reasoning component by component.
\end{proof}
\begin{maremarque}
    The element $\zeta_l$ can be efficiently computed using the following formula from the proof of the Hilbert 90 Theorem: it can be chosen as the multiplicative inverse of any nonzero element in the image of the endomorphism $\sum_{0\leq i < r} \prod_{0\leq j < i} \Frob^j(z_l)\Frob^i$.
\end{maremarque}

\begin{madefinition}
For $f \in \mathbf{E}_{k}^{(l)}$, we set
${f}^\bullet := (\zeta_l f \zeta_l^{-1})^* = \zeta_{\tau(l)}^{-1} f^* \zeta_{\tau(l)} \in \mathbf{E}_{k}^{(\tau(l))}$.
\end{madefinition}

\begin{monlemme}
For all $f \in \tilde{\mathbf{E}}_{k}^{(l)}$, we have $f^*(x_{\tau(l)} X) = f(x_l X)^\bullet$.
\end{monlemme}

\begin{proof}
By additivity, it is enough to check the formula when $f$ is the monomial $\kappa X^i$.
We thus have $f^* = X^{-i}{{\sigma}}_{l}(\kappa)$ and so 
\begin{align*}
f^*(x_{\tau(l)} X) 
 & = X^{-i} {{\sigma}}_{l}(\kappa) \left(\prod_{t=0}^{i-1} {\Frob^t}(x_{\tau(l)})\right)^{-1} \\
f(x_l X)^\bullet
 & = \left(\kappa \prod_{t=0}^{i-1} \Frob^t(x_l) X^{i}\right)^\bullet
   = (z_{\tau(l)} X)^{-i}{{\sigma}}_{l}\left(\kappa \prod_{t=0}^{i-1} \Frob^t(x_{l})\right).
\end{align*}
We conclude by noticing that
$\prod_{t=0}^{i-1} \Frob^t(z_{\tau(l)})=\prod_{t=0}^{i-1} \Frob^t(x_{\tau(l)}){{\sigma}}_{l}\left(\prod_{t=0}^{i-1} \Frob^t(x_{l})\right)$.
\end{proof}

Following the isomorphism $\mathbf{E}_{k}^{(l)} \simeq \End{\mathbf F_l}{\mathbf K_l}$ and its counterpart for $\tau(l)$, we find that the adjunction $f \mapsto f^\bullet$ induces another anti-isomorphism
$\End{\mathbf{F}_l}{\mathbf{K}_l} \stackrel\bullet\longrightarrow \End{\mathbf{F}_{\tau(l)}}{\mathbf{K}_{\tau(l)}}$.
We are now going to prove that the latter is the adjunction map associated to some explicit bilinear map. Precisely, we introduce the twisted bilinear trace form
\begin{equation}
\label{eq:formonKl}
\begin{array}{rcl}
  \mathbf{K}_l \times \mathbf{K}_{\tau(l)} & \longrightarrow & \mathbf{F}_l\\
  (\kappa,\rho) &\mapsto & (\kappa,\rho)_{{\mathbf{F}_l}} := \Trace {\mathbf{K}_l}{\mathbf{F}_l}{\zeta_{l}\cdot \kappa\cdot{{\sigma}}_{\tau(l)}(\rho)}
\end{array}
\end{equation}
In the palindromic case, we have $\mathbf{K}_{\tau(l)} = \mathbf{K}_l$ and we observe that the above pairing is Euclidean when $y_l = \pm 1$ and Hermitian otherwise.
In all cases, the bilinear form $(-,-)_{{\mathbf{F}_l}}$ is nondegenerate and hence identifies $\mathbf{K}_l$ with the dual of $\mathbf{K}_{\tau(l)}$.

\begin{maproposition} The involutive isomorphism $\bullet$ is the adjunction relative to $(-,-)_{{\mathbf{F}_l}}$, \emph{i.e.}
$$(f(\kappa),\rho)_{{\mathbf{F}_l}} = (\kappa,f^\bullet(\rho))_{{\mathbf{F}_l}}, \qquad
  \forall f \in \End{\mathbf{F}_l}{\mathbf{K}_l}, \quad
  \forall \kappa \in \mathbf{K}_l, \quad
  \forall \rho \in \mathbf{K}_{\tau(l)}.$$
\end{maproposition}
\begin{proof}
We write $f = \sum_{0\leq i \leq r-1}f_i \Frobl^i$ with $f_i \in \mathbf K_l$ and compute
\begin{align*}
 (f(\kappa),\rho)_{{\mathbf{F}_l}} 
   & = \sum_{k=0}^{r-1} \Frobl^{k}\left(\zeta_{l}\cdot {{\sigma}}_{\tau(l)}(\rho) \cdot \sum_{i=0}^{r-1}  f_{i}\Frobl^{i}(\kappa)\right)  \\
   & = \sum_{i=0}^{r-1} \sum_{k=0}^{r-1}  \Frobl^{k+i}\left(\zeta_{l}\cdot \Frobl^{-i}(f_{i}) \cdot \frac{\Frobl^{-i}(\zeta_{l})}{\zeta_{l}} \cdot \Frobl^{-i}({{\sigma}}_{\tau(l)}(\rho)) \cdot \kappa\right) \\
   & = \sum_{k=0}^{r-1} \Frobl^{k}\left(\sum_{i=0}^{r-1} \zeta_{l}\cdot \Frobl^{-i}(f_{i}) \cdot \Frobl^{-i}(\zeta_l {{\sigma}}_{\tau(l)}(\rho)) \cdot \zeta_l^{-1} \cdot \kappa\right) \\
   & = \Trace {\mathbf{K}_l}{\mathbf{F}_l}{\zeta_{l} \cdot {{\sigma}}_{\tau(l)}(f^{\bullet}(\rho)) \cdot \kappa} = (\kappa,f^{\bullet}(\rho))_{{\mathbf{F}_l}} 
\end{align*}
which is exactly what we want.
\end{proof}
Finally, composing the morphisms $X \mapsto x_l X$ and $X \mapsto \Frobl$, we obtain the following commutative diagram:
\begin{equation}
\label{diagrambullet}
\begin{tikzcd}
    { \tilde{\mathbf{E}}_{k}^{(l)}} \arrow[dd, "{f \mapsto {f^{*}}}"] \arrow[rr, "X \mapsto x_l X"] && \mathbf{E}_{k}^{(l)} \arrow[dd, "{f \mapsto {f^\bullet}}"] \arrow[rr, "X \mapsto \Frobl"] && \End{\mathbf{F}_l}{\mathbf{K}_l}  \arrow[dd, "{\text{adjunction for }(-,-)_{\mathbf F_l}}"]\\[-2ex]
    & \\[-2ex]
    { \tilde{\mathbf{E}}_{k}^{(\tau(l))}}  \arrow[rr, "X \mapsto x_{\tau(l)} X"]  && \mathbf{E}_{k}^{(\tau(l))} \arrow[rr, "X \mapsto \Frobtl"]  && \End{\mathbf{F}_{\tau(l)}}{\mathbf{K}_{\tau(l)}}
\end{tikzcd}
\end{equation}
where we note that the composite of the horizontal maps is $\mathcal E_l$ on the top, and $\mathcal E_{\tau(l)}$ on the bottom.

\subsection{Vector space duality} \label{section2}

In the previous subsections, we reduced the problem of finding selfdual skew cyclic codes in $\mathbf{E}_{k}$ to that of finding selfdual skew cyclic codes in the product of the $\End{\mathbf{F}_{l}}{\mathbf{K}_{l}}$.
We will now further reduce this problem to that of finding maximal isotropic $\mathbf{F}_{l}$-vector spaces of $\mathbf{K}_{l}$ in the palindromic case and of $\mathbf{K}_{l}\times \mathbf{K}_{\tau(l)}$ in the nonpalindromic case.

To this end, we apply the classical duality between $\mathbf{F}_l$-vector subspaces of $\mathbf{K}_l$ and left ideals of $\End{\mathbf{F}_l}{\mathbf{K}_l}$.
Let us recall it briefly.
Given a field $F$ and a finite dimensional $F$-vector space $W$, the \textit{vector space duality} associates to every $F$-vector subspace $V$ of $W$,
the left ideal $I_V$ of $\End{F}{W}$ formed by the endomorphisms vanishing on $V$. Dually, it associates to every left ideal $I$ of $\End{F}{W}$, the intersection of the kernels of the morphisms in $I$.
With formulas, it can be expressed as
\begin{align*}
  I \mapsto V_I & = \bigcap_{f\in I} \ker(f), \\
  V \mapsto I_V & = \big\{\, f \in \End F W \,|\, V \subset \ker(f) \,\big\}.
\end{align*}
This duality defines an order-reversing one-to-one correspondence between the set of left ideals of $\End{F}{W}$ and the set of $F$-vector subspaces of $W$.
Moreover, for all $V \subset W$, we have $\dim_F I_V = (\dim_F W - \dim_F V) \cdot \dim_F W$.

We now assume in addition that we are given an involution $\sigma : F \to F$ and that $W$ is endowed with a nondegenerate $\sigma$-sesquilinear form.
We recall that this datum equips $\End F W$ with a sesquilinear form as well. In particular, taking orthogonals over $W$ and $\End F W$ makes sense.
\begin{maproposition}
For all subspace $V$ of $W$, we have ${\Ortho{I_{V}}=I_{\Ortho{V}}}$.
\end{maproposition}
\begin{proof}
Given $f \in I_V$ and $g \in I_{\Ortho{V}}$, we have
$f \circ \Adjoint{g} = 0$ since $f$ vanishes on $V$ and $\mathrm{im}\:\Adjoint{g} = \Ortho{(\ker g)} \subset V$.
Therefore $f$ and $g$ are orthogonal in $\End F W$. It follows that ${\Ortho{I_{V}} \subset I_{\Ortho{V}}}$.
The equality follows by comparing dimensions.
\end{proof}

We are now ready to apply what precedes to codes and prove the main theorem of this section.
\begin{montheorem} \label{globmainprop} There exists an explicit bijection between the set of selfdual skew cyclic codes of $\mathbf{E}_{k}$ and the cartesian product of sets $W_{\mathtt{pal}}\times W_{\mathtt{nonpal}}$, where:
    \begin{itemize}
        \item $W_{\mathtt{pal}}$ is the cartesian product, over the set $I$ of indexes invariant under $\tau$, of the sets of isotropic $\mathbf{F}_l$-vector subspaces of $\mathbf{K}_l$ of dimension $r/2$,
        \item $W_{\mathtt{nonpal}}$ is the cartesian product, over the set $J$ of all remaining nontrivial orbits of $\tau$, of the sets of $\mathbf{F}_l$-vector subspaces of $\mathbf{K}_l$.
    \end{itemize}
\end{montheorem}
\begin{proof}
By what we have done in previous subsections, selfdual codes in $\mathbf E_k$ are in bijection with left ideals of the cartesian product
$$\prod_{l=1}^n \End {\mathbf F_l}{\mathbf K_l}$$
that are equal to their orthogonal.
Besides, the orthogonal of an ideal can be taken component by component, with the care that the orthogonal of the $l$-th component lies in the $\tau(l)$-th component.
Therefore, when $\tau(l) = l$, the $l$-th component must be selfdual itself whereas, when $\tau(l) \neq l$, the component at position $l$ can be anything but it determines the component at position $\tau(l)$.
Using now the vector space duality, we can further replace ideals of $\End {\mathbf F_l}{\mathbf K_l}$ by $\mathbf F_l$-subspaces of $\mathbf K_l$.
This operation preserves the orthogonality condition as the vector space duality commutes with orthogonals.

We finally conclude by noticing that a subspace of $\mathbf K_l$ which is equal to its orthogonal is nothing else than an isotropic subspace of half dimension, that is of dimension $r/2$.
\end{proof}

\section{Counting and generating selfdual skew cyclic codes}\label{section:algo}

We keep the notation introduced before. In particular, we recall that $\mathbf K/\mathbf F$ is an extension of finite fields of degree $r$ and that $\mathbf E_k = \mathbf K[X;\Frob]/(X^{kr} - 1)$ (where $\Frob : x \mapsto x^q$ with $q = \text{Card } \mathbf F$).
Besides, we set $Y = X^r$ and assume that $k$ is coprime with $r$. Under this hypothesis, the polynomial $Y^k - 1$ is separable and we write down its decomposition as a product of irreductible factors $Y^k - 1 = P_1(Y) \cdots P_n(Y)$.
We recall also that we have introduced an involution $\tau : \{1, \ldots, n\} \to \{1, \ldots, n\}$ defined by the condition that the roots of $P_l$ are the inverses of the roots of $P_{\tau(l)}$.
In Subsection~\ref{ssec:evalisom}, we proved that we have an isomorphism of the form
$$\mathbf E_k \simeq \prod_{l=1}^s \mathbf K_l[X;\Frob]/(X^r - y_l) \simeq \prod_{l=1}^s \End{\mathbf F_l}{\mathbf K_l}$$
where $\mathbf F_l = \mathbf F[Y]/P_l(Y)$, $\mathbf K_l = \mathbf K \otimes_{\mathbf F} \mathbf F_l = \mathbf K[Y]/P_l(Y)$ and $y_l$ is the image of $Y$ in $\mathbf K_l$.
In Subsection~\ref{ssec:adjunction}, we showed that this decomposition preserves orthogonality in some precise sense.
This allowed us to conclude (see Theorem~\ref{globmainprop}) that enumerating selfdual skew cyclic codes sitting in $\mathbf E_k$ boils down to enumerating maximal isotropic $\mathbf{F}_l$-vector subspaces of $\mathbf{K}_l$ when $\tau(l) = l$ (palindromic case), and to enumerating $\mathbf{F}_l$-vector subspaces of $\mathbf{K}_l$ otherwise.

In this section, we rely on this theoretical result, first, to count skew cyclic codes and, second, to construct them explicitely. 
More precisely, we shall address two different problems: that of random generation and that of complete enumeration.

\emph{Throughout this section, we assume that the characteristic of $\mathbf F$ is odd.}

\subsection{Existence criterion}

By Theorem~\ref{globmainprop}, there exist selfdual codes in $\mathbf E_k$ if and only if for each $l$ such that $\tau(l) = l$, the space $\mathbf K_l$ admits a totally isotropic subspace of dimension $s := r/2$.
We then aim at providing simpler conditions for this property to hold.
For this, we shall use Witt's decomposition theorem as a fundamental tool. Let us recall it briefly.
Let $F$ be a field of odd characteristic, and let $\sigma : F \to F$ be a ring homomorphism which is an involution (possibly the identity). Let $V$ be a finite dimension vector space over $F$, endowed with a $\sigma$-sesquilinear form $\mathcal B : V \times V \to F$.
We recall that a \textit{hyperbolic pair} is a pair of vectors $(u, v)$ of $V$ satisfying $\mathcal{B}(u,u) = 0$, $\mathcal{B}(v,v) = 0$ and $\mathcal{B}(u, v) = 1$, and that the $2$-dimensional subspace of $V$ spanned by a hyperbolic pair $(u, v)$ is called a \textit{hyperbolic plane}.

\begin{montheorem}\label{WittDec}
Keeping the previous notation, there exists an invariant $d$ (called the Witt index of $V$) and hyperbolic planes $H_1, \ldots, H_d$ such that one has the orthogonal decomposition
$$ V \simeq \left(\bigoplus_{1 \leq i \leq d} H_{i}\right) \oplus W $$
where $W$ is a subspace that does not contain any nonzero isotropic vector.

Moreover, the dimension of any maximal isotropic space is equal to $d$.
\end{montheorem}

\begin{proof}
See for instance \cite[Theorem 3.11]{21}.
\end{proof}

When $F$ is a finite field, more can be said. For simplicity, we assume that $\dim V = 2s$.
If $\sigma \neq \Id$, the Witt index of $V$ is always $s$. On the contrary, when $\sigma = \Id$, it can be either $s$ or $s{-}1$ but we can decide between those two values by looking at the discriminant $\delta_V$ of $V$ (defined as the determinant of the matrix of $\mathcal B$ in some basis); precisely, the Witt index is $s$ if and only if $(-1)^s \delta_V$ is a square in $F^\times$.
(See \cite[Theorem 3.3]{22} for more details.)

In our case, Theorem~\ref{globmainprop} tells us that we are looking for isotropic vectors of dimension $s$ in $\mathbf K_l$; we recall from Equation~\eqref{eq:formonKl} that the latter is endowed with the sesquilinear form
$$(\kappa,\rho)_{{\mathbf{F}_l}} = \Trace {\mathbf{K}_l}{\mathbf{F}_l}{\zeta_{l}\cdot \kappa\cdot{{\sigma}}_{\tau(l)}(\rho)}$$
where $\sigma_{\tau(l)} : \mathbf K_{\tau(l)} \to \mathbf K_l$ is the map induced by $\sigma_{\tau(l)}(Y) = \frac 1 Y$ and $\zeta_l$ is an element of $\mathbf K_l$ defined in Lemma~\ref{lem:zetal}.
We then need to compute the discriminant $\delta_{\zeta_l}$ of this sesquilinear form.

\begin{monlemme}\label{normdisc}
We assume that $\sigma_l = \Id$ and we let $\delta_{\mathbf{K}_l/\mathbf{F}_l}$ be the discriminant of the extension $\mathbf K_l/\mathbf F_l$ (which is, by definition, the discriminant of the bilinear form $(\kappa,\rho) \mapsto \Trace {\mathbf{K}_l}{\mathbf{F}_l}{\kappa \rho}$). Then
    \begin{enumerate}
    \item the discriminant $\delta_{\zeta_l}$ is equal to $\Norm{\mathbf{K}_l}{\mathbf{F}_l}{\zeta_{l}} \cdot \delta_{\mathbf{K}_l/\mathbf{F}_l}$,
    \item the discriminant $\delta_{\mathbf{K}_l/\mathbf{F}_l}$ is a square in $\mathbf{F}_l$ if and only if the degree of the extension $[\mathbf{F}_l:\mathbf{F}]$ is even,
    \item if $y_l=1$ (resp. $y_l=-1$), $\Norm{\mathbf{K}_l}{\mathbf{F}_l}{\zeta_{l}}$ is a square (resp. is not a square) in $\mathbf{F}_l$.
    \end{enumerate}
\end{monlemme}
\begin{proof}
\noindent
1. We fix a basis of $\mathbf K_l$ over $\mathbf F_l$ and write $\textrm{Mat}(\zeta_l)$ for the matrix representing the multiplication by $\zeta_{l}$ in this basis. Then
$\delta_{\zeta_l}= \det{({{\textrm{Mat}(\zeta_l)}^{\mathsf{tr}}})} \cdot \delta_{\mathbf{K}_l/\mathbf{F}_l} = \Norm{\mathbf{K}_l}{\mathbf{F}_l}{\zeta_{l}} \cdot  \delta_{\mathbf{K}_l/\mathbf{F}_l}$.

\noindent
2. From $\mathbf K_l = \mathbf K \otimes_{\mathbf F} \mathbf F_l$, we deduce that $\delta_{\mathbf{K}_l/\mathbf{F}_l} = \delta_{\mathbf{K}/\mathbf{F}} \in \mathbf F$.
Moreover, we know that $\delta_{\mathbf{K}/\mathbf{F}}$ is a square in $\mathbf{F}$ if and only if the Galois group of $\mathbf{K}/\mathbf{F}$ is a subgroup of the alternating group (see~\cite[Corollary 4.2]{14}), which never occurs in our situation given that $\text{Gal}(\mathbf{K}/\mathbf{F})$ is a cyclic group of even cardinality.
We conclude that $\delta_{\mathbf{K}/\mathbf{F}}$ is a square in $\mathbf{F}_l$ if and only if the extension $\mathbf{F}_l/\mathbf{F}$ has even degree.

\noindent
3. We assume that $y_l = \pm 1$ and compute 
    \begin{align*}
    {\Norm{\mathbf{K}_l}{\mathbf{F}_l}{\zeta_{l}}}^{\frac{q-1}{2}} 
      & = \left(\zeta_l^{\sum_{0 \leq i < 2s} q^i}\right)^{\frac{q-1}{2}} \\
      & = \big(\zeta_l^{q-1}\big)^{\frac{\sum_{0 \leq i < 2s} q^i}{2}}= \big(x_l{{\sigma}}_{l}(x_{l})\big)^{\frac{\sum_{0 \leq i < 2s} q^i}{2}}
    \end{align*}
    As $y_l = \pm 1$, the automorphism $\sigma_l$ is the identity and so ${\Norm{\mathbf{K}_l}{\mathbf{F}_l}{\zeta_{l}}}^{\frac{q-1}{2}} = y_l$.
    We conclude by applying Euler's criterion.
\end{proof}

\begin{moncorollaire}\label{cor:existcrit}
We assume that the characteristic of $\mathbf F$ is odd.
\begin{enumerate}
\item If $k$ is even, there are no selfdual skew cyclic codes in $\mathbf{E}_k$.
\item If $k$ is odd, there exist selfdual skew cyclic codes in $\mathbf{E}_k$ if and only if $(-1)^s$ is not a square in $\mathbf F$, if and only if $s$ is odd and $q \equiv 3 \pmod 4$.
\end{enumerate}
\end{moncorollaire}
\begin{proof}
We first notice that, whenever $y_l \neq \pm 1$, there is no obstruction to the existence of an isotropic subspace of half dimension.
On the contrary, when $y_l = 1$ (resp. $y_l = -1$), it follows from Lemma~\ref{normdisc} that an isotropic subspace of $\mathbf K$ of dimension $s$ exists if and only if $(-1)^s$ is a square (resp. is not a square) in $\mathbf F$.

When $k$ is even, the decomposition of $\mathbf{E}_{k}$ exhibits both factors $\mathbf{K}[X;\Frob]/(X^{r}+1)$ and $\mathbf{K}[X;\Frob]/(X^{r}-1)$.
Since $(-1)^s$ cannot be simultaneously a square and a nonsquare, we conclude that selfdual skew cyclic codes cannot exist in this case.
On the contrary, when $k$ is odd, the factor $\mathbf{K}[X;\Frob]/(X^{r}+1)$ does not show up and we are left to the condition corresponding to $y_l = 1$.

Finally, the fact that if $(-1)^s$ is not a square in $\mathbf F$ if and only if $s$ is odd and $q \equiv 3 \pmod 4$ is a direct application of Euler's criterion.
\end{proof}

\subsection{Counting selfdual skew cyclic codes}

We now aim at counting the number of selfdual codes sitting in $\mathbf{E}_k$, when they exist.
In what follows, we then assume that the existence criterion of Corollary~\ref{cor:existcrit} is fulfilled.
It follows from Theorem~\ref{globmainprop} that out task reduces to finding the cardinality of $W_{\mathtt{pal}}$ and $W_{\mathtt{nonpal}}$.

\subsubsection{The nonpalindromic case}

We start by the nonpalindromic case, which is by far the easiest.
For this counting, we will use $q$-analogues of integers.
We recall briefly that the $q$-analogue of $n \in \mathbf N$ is, by definition, $[n]_q:= 1+q+q^2+\cdots+q^{n-1}$.
The $q$-factorial of $n$ is defined by $[n]_q!=[1]_q[2]_q\ldots[n]_q$ and we set
$${\qbin{n}{k}{q}} := 
  \frac{[n]_q!}{[k]_q!\,[n-k]_q!} =
  \frac{(1-q^n)(1-q^{n-1})\ldots(1-q^{n-k+1})}{(1-q)(1-q^2)\ldots(1-q^k)}$$
where $n$ and $k$ are nonnegative integers with $k \leq n$.
It is a classical fact that the $q$-binomial coefficients count the number of $\mathbb{F}_q$-vector subspaces of dimension $k$ in the ambient $\mathbb{F}_q$-vector space $\mathbb{F}_q^n$.

Therefore, with the notation of Theorem~\ref{globmainprop}, we have
\begin{equation}
\label{eq:Wnonpal}
\text{Card}(W_{\mathtt{nonpal}}) = 
  \prod_{\{l, \tau(l)\}\in J} \left( \sum_{k=0}^{r} \qbin r k {q_l} \right).
\end{equation}

\subsubsection{The palindromic case}

If $V$ is a finite dimensional vector space equipped with a sesquilinear form, we denote by $\Iso(V)$ the number of isotropic subspaces of $V$ of half dimension.
It turns out that the behaviour of $\Iso(V)$ have been studied for a long time (see for instance \cite{2,7,16}) and that explicit formulas are known. Those are called Segre's formulas and are recalled in the following theorem.

\begin{montheoreme}\label{SEGRE}
Let $F$ be a finite field of odd characteristic and cardinality $q_F$ and
let $\sigma : F \to F$ be an involutive ring automorphism.
Let $V$ be a $F$-vector space of dimension $2s$ equipped with a nondegenerate
$\sigma$-sesquilinear form, whose Witt index is $s$.
Then:
\begin{enumerate}
\item if $\sigma = \Id$ (Euclidean case), then
$\Iso(V) = \prod\limits_{i=0}^{s-1}\left(q_F^{i}+1\right)$,
\item if $\sigma \neq \Id$ (Hermitian case), then
$\Iso(V) = \prod\limits_{i=0}^{s-1}\left(q_F^{i+1/2}+1\right)$.
\end{enumerate}
\end{montheoreme}

\begin{proof}
We recall briefly the idea of the proof as it will be useful afterwards.
Let $\IsoVect(W)$ be the number of isotropic vectors in a Euclidean or Hermitian vector space $W$ over $F$. We claim that, if $W$ has dimension $2d$ and Witt index $d$, then
\begin{enumerate}
\item if $\sigma = \Id$ (Euclidean case), then
$\IsoVect(W) = (q_F^{d}-1)(q_F^{d-1}+1)$,
\item if $\sigma \neq \Id$ (Hermitian case), then
$\IsoVect(W) = (q_F^{d}-1)(q_F^{d-1/2}+1)$
\end{enumerate}
Indeed, let us fix an isotropic basis $((u_i)_{0 \leq i<s},(v_i)_{0 \leq i<s})$ corresponding to the Witt's decomposition of $W$ (see Theorem~\ref{WittDec}) and let $((a_i)_{0 \leq i < s},(b_i)_{0 \leq i < s})$ be the coordinates in this basis of a vector.
In the Euclidean case, the fact that this vector is isotropic reduces to the equation $\sum_{0 \leq i < s} a_i b_i =0$.
Now, fixing a nonzero vector $(a_i)_{0 \leq i < s}$ of $\mathbf{F}_l^s$, this occurs if and only if $(b_i)_{0 \leq i < s}$ lies in some hyperplane.
We thus have $(q_F^{s}-1)q_F^{s-1}$ solutions corresponding to nonzero $(a_i)_{0 \leq i < s}$, to which one should add $(q_F^{s}-1)$ more solutions when all $a_i$ vanish.
Finally, we get $\IsoVect(W)=(q_F^{d}-1)(q_F^{d-1}+1)$ as claimed.

The Hermitian case is similar, expect that the equation to solve is now $\sum_{0 \leq i < s} \sigma_{l}(a_i)b_i +\sum a_i\sigma_{l}(b_i)=0$, which reduces to
$\sum_{0 \leq i < s} a_i\sigma_{l}(b_i)=\alpha$ where $\alpha$ satisfies $\sigma_{l}(\alpha)=-\alpha$.
We conclude repeating the argument of the Euclidean case and using that there are exacly $q_F^{1/2}$ values for $\alpha$.

We are now ready to prove Segre's formula. We start by taking $W_1 = V$ and by picking an isotropic vector $u_1$ in $W$. This corresponds to $\IsoVect(W_1)$ possibilities.
Once this is achieved, we set $W_2 := (Fu_1)^\perp / Fu_1$. This is a space of dimension $2(s{-}1)$, whose Witt index is $s{-}1$. 
Therefore, we can apply again our claim and find that there are exactly $\IsoVect(W_2)$ isotropic vectors in $W_2$. We choose one of them, that we denote by $u_2$. 
Now we repeat the argument until we reach $u_s$.
This corresponds to $\IsoVect(W_1) \cdot \IsoVect(W_2) \cdots \IsoVect(W_s)$ choices. 
However, each of them corresponds $q_F \cdot q_F^2 \cdots q_F^{s-1} = q_F^{s(s-1)/2}$ choices of families of vecteurs of $V$ since $u_i$ has $q^i$ preimages in $V$.
We conclude that the number of bases of a maximal isotropic subspace of $V$ is equal to $q_F^{s(s-1)/2} \cdot \IsoVect(W_1) \cdot \IsoVect(W_2) \cdots \IsoVect(W_s)$.
We finally obtain $\Iso(V)$ by dividing by the cardinality of $\GL s F$.
\end{proof}

\begin{maremarque}
\label{rem:segre}
In the Euclidean case, one can alternatively prove Segre's formula by remarking that the orthogonal group of $V$ acts transitively on the set of maximal isotropic subspaces
and that the stabilizer of a given maximal isotropic subspace $U$ can be presented as a semi-direct product of $\text{GL}(U)$ and the group of antisymmetric linear applications from $U$ to its dual $\text{Hom}_F(U,F)$.
From this description we find that the number of maximal isotropic subspaces is
$$\frac{\text{Card}\big(O_{2s}(F)\big)}{q_F^{s(s-1)/2} \cdot \text{Card}\big(\GL s F\big)}$$
a formula from which one can eventually derive Segre's theorem.
A similar approach also works in the Hermitian case.
\end{maremarque}

Keeping the notation of Theorem~\ref{globmainprop}, it follows from Theorem~\ref{SEGRE} that
\begin{equation}
\label{eq:Wpal}
\text{Card}(W_{\mathtt{pal}}) =
  \prod_{\substack{l \in I\\ y_l = \pm 1}}  \prod\limits_{i=0}^{s-1}\left(q_l^{i}+1\right) \times
  \prod_{\substack{l\in I \\ y_l\neq \pm 1}}\prod\limits_{i=0}^{s-1}\left(q_l^{i+1/2}+1\right).
\end{equation}
We notice moreover that there is always exactly one index $l$ for which $y_l = 1$, and there is at most one index $l$ such that $y_l = -1$ (such an index actually exists if and only if $k$ is even).
In both cases, the corresponding field $\mathbf F_l$ is $\mathbf F$, and so $q_l = q$.

Now combining Equations~\eqref{eq:Wnonpal} and~\eqref{eq:Wpal}, we get the number of selfdual skew cyclic codes sitting in $\mathbf E_k$, which proves Theorem~\ref{maincount}.

\begin{monexample}
    For $\mathbf{K}=\mathbb{F}_{q^{2s}}$ and $\Frob: x \mapsto x^q$, the number of selfdual skew cyclic codes is equivalent to $q^{\frac{s(s-1)}{2}}$ as $s$ grows to infinity, whereas the number of skew cyclic codes (number of $s$ dimensional $\mathbb{F}_{q}$-vector subspaces of $\mathbb{F}_{q^{2s}}$) is equivalent to $q^{s^2}$ as $s$ grows to infinity.

    For example, for $\mathbf{K}=\mathbb{F}_{3^6}$ and $\Frob: x \mapsto x^3$, the number of selfdual skew cyclic codes in $\mathbf{E}_{1}=\mathbf{K}[X;\Frob]/(X^6-1)$ is $80$ among  $33880$ skew cyclic codes,
    whereas for $\mathbf{K}=\mathbb{F}_{3^{18}}$ and $\Frob: x \mapsto x^3$, the number of selfdual skew cyclic codes in $\mathbf{E}_{1}=\mathbf{K}[X;\Frob]/(X^{18}-1)$ is $469740602936729600$ among  $791614563787525746761491781638123230424$ skew cyclic codes.
\end{monexample}
\begin{maremarque}
    As we noticed at the end of subsection~\ref{section01}, there exist no selfdual skew cyclic codes at all in the case $r$ odd.
    This is in particular the case for $r=1$, when skew cyclic codes reduce to cyclic codes.
    On the contrary, for $s=1$ and $q=3$, we achieve the best ratio of selfdual skew cyclic codes over skew cyclic codes, which is $\frac{1}{2}$.
    Nevertheless, the proportion of selfdual skew cyclic codes over skew cyclic codes decreases as fast as 
    $\mathcal{O}(q^{-\frac{s^2+s}{2}})$ as $s$ grows larger.  
\end{maremarque}

\subsection{Random generation of selfdual skew cyclic codes}

Since the number of selfdual skew cyclic codes grows exponentially fast with respect to the dimension~$r$, an algorithm outputting in one shot the complete list of these codes would be necessarily very unefficient (the better we can expect is exponential complexity) and hence, probably not quite useful.
Instead, in what follows, we address a different question, which is that of random generation: we aim at finding a fast algorithm that outputs a unique code in this huge list with the guarantee that the returned code is \emph{uniformly distributed} among all of them. Such an algorithm could be very useful to generate \emph{typical} selfdual skew cyclic codes and to check their properties.

\subsubsection{From skew cyclic codes to finite geometry: explicit methods}

Before designing our algorithms, we need to explain how we represent the objects on the computer.
We recall that a skew cyclic code is, by definition, a left ideal of $\mathbf E_k = \mathbf K[X; \Frob]/(X^{kr} - 1)$. Hence it necessarily has the form $\mathbf E_k f$ for some $f \in \mathbf K[X;\Frob]$.
We can further normalize this generator by requiring that it is monic and has minimal degree; concretely, normalizing the generator $f$ amounts to replacing it by $\rgcd(f, X^{kr}{-}1)$ (where $\rgcd$ denotes the right gcd).
The same discussion applies similarly to all quotients of a Ore polynomial ring by a two-sided ideal and so, in particular, to $\mathbf K[X; \Frob]/P_l(Y)$ and the algebras $\tilde{\mathbf E}_k^{(l)} = \mathbf K_l[X; \Frob]/(X^r - y_l)$.

We recall further that we have the following sequence of isomorphisms:
$$\mathbf E_k \simeq \prod_{l=1}^n \mathbf K[X; \Frob]/P_l(Y) 
\simeq \prod_{l=1}^n \tilde {\mathbf E}_k^{(l)} \simeq \prod_{l=1}^n \End {\mathbf F_l}{\mathbf K_l}$$
and that the left ideals of $\End {\mathbf F_l}{\mathbf K_l}$ are in one-to-one correspondence with the $\mathbf F_l$-linear subspaces of $\mathbf K_l$ (see Subsection~\ref{section2}).
We aim at making explicit all these identifications.

Going back and forth between $\mathbf E_k$ and $\prod_{l=1}^n \tilde {\mathbf E}_k$ is not difficult. Indeed, if a code sitting in $\mathbf E_k$ is generated by $f$, its image in $\mathbf E_k$ will be generated by $f$ as well.
Conversely, if one starts with a family of codes $\big(\tilde {\mathbf E}_k^{(l)} f_l\big)_{1 \leq l \leq n}$, its preimage in $\mathbf E_k$ is the code generated by a Ore polynomial $f$ satisfying the set of congruences 
\begin{equation}
\label{eq:CRT}
f \equiv f_l \pmod{P_l(Y)} \qquad (1 \leq l \leq n).
\end{equation}
We need to be careful however that $f_l$ has \emph{a priori} coefficients in $\mathbf K_l$; in order to view it as a Ore polynomial in $\mathbf K[X; \Frob]$, we have to replace each occurrence of $y_l$ by $Y = X^r$.
The system of congruences~\eqref{eq:CRT} can then be solved using the Chinese Remainder Theorem; we underline that noncommutativity is not an issue here because all the moduli $P_l(Y)$ lie in the center.
We also stress that the solution $f$ to \eqref{eq:CRT} is in general not normalized, even if the $f_l$ are; if one wants to normalize it, one needs to compute an additional $\rgcd$.

We now explain how to navigate between $\tilde {\mathbf E}_k^{(l)}$ and $\End {\mathbf F_l}{\mathbf K_l}$.
We first recall that the isomorphism between those two rings is given by $X \mapsto x_i \Frob$.
Hence the ideal of $\End {\mathbf F_l}{\mathbf K_l}$ that corresponds to the ideal $\tilde {\mathbf E}_k^{(l)} f$ of $\tilde {\mathbf E}_k^{(l)}$ is the ideal consisting on linear maps vanishing on the kernel of $f(x_i \theta)$.
The associated $\mathbf F_l$-linear subspace of $\mathbf K_l$ is then just $\ker f(x_l \theta)$.
The correspondence in the other direction is also given by an explicit formula: if $V$ is a $\mathbf F_l$-subvector space of $\mathbf K_l$ and $(v_1, \ldots, v_d)$ is a basis of $V$, a generator of the ideal of $\tilde {\mathbf E}_k^{(l)}$ corresponding to $V$ is
$$\llcm\left(X - \frac{x_l \Frob(v_1)}{v_1}, \, \ldots,\, X - \frac{x_l \Frob(v_d)}{v_d}\right)$$
where $\llcm$ denotes the left lcm.

The discussion of this subsection is summarized by Algorithm~\ref{algo:bijection} below which computes the normalized generator of the code sitting in $\mathbf E_k$ that corresponds to some element of $W_{\mathtt{nonpal}} \times W_{\mathtt{pal}}$ \emph{via} the bijection of Theorem~\ref{globmainprop}.
The next subsections are devoted to explain how to produce a random element in (each component of) $W_{\mathtt{nonpal}} \times W_{\mathtt{pal}}$.

\begin{algorithm}
    \caption{Explicit bijection with $W_{\mathtt{nonpal}} \times W_{\mathtt{pal}}$}\label{algo:bijection}
    \begin{algorithmic}[1]
            \Require{a family $\big((V_l)_{l \in I}, (V_l)_{\{l, \tau(l)\} \in J}\big) \in W_{\mathtt{nonpal}} \times W_{\mathtt{pal}}$}
            \Ensure{the normalized generator of the corresponding selfdual skew cyclic code}
            \State \ForCustom{$l \in I$}{
                \State pick a basis $(v_1, \ldots, v_s)$ of $V_l$
                \State $f_l \gets \llcm\big(X - x_l\Frob(v_i)/v_i, \, 1 \leq i \leq s\big)$
                \State do the subtitution $y_l \to X^r$ in $f_l$ \hfilll \Comment{now $f_l \in \mathbf K[X;\theta]$}
                \State $f_l \gets \rgcd(f_l, P_l(Y))$
            }
            \State \ForCustom{$\{l, \tau(l)\} \in J$}{
                \State pick a basis $(v_1, \ldots, v_d)$ of $V_l$
                \State $f_l \gets \llcm\big(X - x_l\Frob(v_i)/v_i, \, 1 \leq i \leq d\big)$
                \State do the subtitution $y_l \to X^r$ in $f_l$ \hfilll \Comment{now $f_l \in \mathbf K[X;\theta]$}
                \State $f_l \gets \rgcd(f_l, P_l(Y))$
                \State define $f_{\tau(l)}$ by the equality $f_l f_{\tau(l)}^* = P_l(Y)$
                \State $f_{\tau(l)} \gets \rgcd(f_{\tau(l)}, P_{\tau(l)}(Y))$
            }
            \State compute $f$ such that $f \equiv f_l \pmod{P_l(Y)}$ for $1 \leq l \leq n$
            \State \Return $\rgcd(f, X^{rk} - 1)$
    \end{algorithmic}
\end{algorithm}

To conclude, we record the following proposition which elucidates how duality acts on our representations.

\begin{maproposition}\label{prodcrit}
We set $E = E' = \mathbf E_k$ and $P = Y^k - 1$ (resp. $E = \mathbf E_k^{(l)}$, $E' = \mathbf E_k^{(\tau(l))}$ and $P = P_l$).
\begin{enumerate}
\item[(a)] Given $f, g \in K[X;\theta]$, the ideals $E f$ and $E' g$ are orthogonal if and only if $f g^* = 0$ in $E$.
\item[(b)] Given $f \in K[X,\theta]$ dividing $P$, the orthogonal of $Ef$ is the ideal $E'g^*$ where $g$ is defined by $fg = P$.
\end{enumerate}
\end{maproposition}

\begin{proof}
(a)~By nondegeneracy of sesquilinear form $\langle-,-\rangle$, the condition $f g^* = 0$ is equivalent to $g f^* = 0$ and then to $\langle E,gf^*\rangle  = 0$.
    By adjunction relation, the condition becomes
    $\langle Ef,g\rangle = 0$.
    Since the adjunction is an isomorphism, the condition is further equivalent to
    $\langle (E')^* Ef,g\rangle = 0$ and finally to
    $\langle Ef, E'g\rangle  = 0$.

\noindent
(b)~By what precedes, the ideals $Ef$ and $E'g^*$ are orthogonal. We conclude by noticing that $\dim Ef + \dim E'g^* = (\deg P - \deg f) + (\deg P - \deg g) = \deg P$. 
\end{proof}

\begin{maremarque}
As a corollary, Proposition~\ref{prodcrit} provides a simple criterion to check that the code $\mathbf E_k f$ is selfdual: assuming that $f$ is normalized, it is the case if and only if $f f^* = 0$ in $\mathbf E_k$ and $\deg f = s$.
\end{maremarque}

\subsubsection{The nonpalindromic case}

We first consider the indices $l$ such that $\tau(l) \neq l$.
At those places, we simply need to generate a uniformly distributed random $\mathbf F_l$-subspace of $\mathbf K_l$.
We proceed as follows. We first construct the dimension: we sample an integer $d \in \{0, \ldots, r\}$ with distribution given by:
$$\text{Prob}[d = i] \quad \text{proportionnal to} \quad \qbin r i {q_l}.$$
Once this is achieved, we sample $d$ random elements in $\mathbf K_l$ with uniform distribution.
If they are linearly independant over $\mathbf F_l$, we output the vector space they generate. Otherwise, we throw them and start again with $d$ new elements.
The probability of failure is
$$\left(1 - \frac 1{q_l^r}\right)
  \left(1 - \frac 1{q_l^{r-1}}\right) \cdots
  \left(1 - \frac 1{q_l^{r-d+1}}\right) \geq
  1 - \left(\frac 1{q_l^r} + \cdots + \frac 1{q_l^{r-d+1}}\right) \geq 1 - \frac{1}{q_l - 1},$$
proving that, in average, we will need to repeat our process only $O(1)$ times.

Up to a multiplicative constant, the mean complexity of the algorithm is then equal to the complexity of checking linearly independence of $d$ vectors in a space of dimension $r$, which is within $O(r^3)$ by Gaussian elimination.

\subsubsection{The Hermitian case}
\label{sssec:random:hermitian}

We now move to the Hermitian case, \emph{i.e.} we assume that $\tau(l) = l$ and $y_l \neq \pm 1$.
We thus want to design an algorithm outputting a uniformly distributed random isotropic $\mathbf F_l$-subspace of $\mathbf K_l$ (endowed with the Hermitian pairing $(-,-)_{\mathbf K_l}$ defined in \eqref{eq:formonKl}, assuming that the existence criterion of Corollary~\ref{cor:existcrit} is fulfilled.
Our construction is inspired by the proof of Theorem~\ref{SEGRE}, except that we will not work with the quotient $\Ortho{(F u_1)}/(F u_1)$ but, instead, will embed $u_1$ is a hyperbolic plane $H_1$ and work with $\Ortho{H_1}$.

We consider a finite field $F$ of cardinality $q_F$ equipped with a nontrivial involutive automorphism $\sigma : F \to F$. We also consider an Hermitian space $V$ of dimension $r$ and denote by $\langle -,-\rangle$ the bilinear form on it. 
We assume that $V$ has Witt index~$s$ (\emph{i.e.} that $V$ is isomorphic to the orthogonal direct sum of $s$ hyperbolic planes) and aim at sampling a random isotropic subspace of $V$ of dimension $s$.

For $u, v \in V$, we consider the following equation in $\lambda$:
$$(\mathcal{E}_{u,v}): \quad \langle u+\lambda v,u+\lambda v \rangle = 0.$$
We briefly recall its resolution.
If $\langle v,v\rangle = 0$, the equation reduces to
$\Trace F {F^\sigma} {\lambda \cdot \langle v, u \rangle} = -\langle u, u \rangle$
which, per surjectivity of the trace, can be solved as soon as $\langle v, u \rangle \neq 0$.

On the contrary, when $\langle v, v \rangle \neq 0$, we consider the \emph{discriminant} of $(\mathcal{E}_{u,v})$ defined by
$\Delta := \langle u,v \rangle \cdot \langle v, u \rangle - \langle u, u \rangle \cdot \langle v, v \rangle$.
One readily checks that $\Delta$ is invariant under $\sigma$ and that the equation $(\mathcal{E}_{u,v})$ can be rewritten
$\Norm F {F^\sigma} {\langle u,v \rangle + \lambda \cdot \langle v, v \rangle} = \Delta$.
The solutions of~$(\mathcal{E}_{u,v})$ are then the elements of the form
$$\lambda = \frac{\delta - \langle u,v \rangle}{\langle v, v \rangle}$$
where $\delta$ is a preimage of $\Delta$ by the norm map.
Since the latter is surjective (because we are working over finite fields), a solution always exists.

We are now ready to present Algorithm~\ref{directsumhermitalgo}: it computes a basis $(u_1, \ldots, u_s, v_1, \ldots, v_s)$ of $V$ such that each pair $(u_i, v_i)$ is hyperbolic and, writing $H_i \subset V$ for the hyperbolic plane they generate, we have the orthogonal decomposition $V = \bigoplus_{i=1}^s H_i$.

\begin{algorithm}
    \caption{Decomposition as a direct sum of hyperbolic planes (Hermitian case)} \label{directsumhermitalgo}
    \begin{algorithmic}[1]
        \Require{$V$: the ambient Hermitian vector space}
        \Ensure{$\mathbf u, \mathbf v$: a basis of hyperbolic pairs}
        \State $\mathbf u, \mathbf v, W\gets [\:],[\:],0$
        \State \WhileCustom{$W \neq V$}{
            \State $u, v \gets $ two random vectors in $\Ortho{W}$
            \State \IfCustom{\rm $(u,v)$ are linearly independent and $\langle v,v \rangle \neq 0$}{
                \State $\lambda \gets $ a random solution of the equation $(\mathcal{E}_{u,v})$
                \State $u \gets u + \lambda v$ \hfilll \Comment{now $\langle u,u \rangle = 0$}
                \State \IfCustom{$\langle u,v \rangle \neq 0$}{
                    \State $\lambda \gets $ a solution of the equation $(\mathcal{E}_{v,u})$
                    \State $v \gets v + \lambda u$ \hfilll \Comment{now $\langle v,v \rangle = 0$}
                    \State $v \gets v / \langle v, u \rangle$ \hfilll \Comment{now $(u,v)$ is a hyperbolic pair}
                    \State $\mathbf u \gets \mathbf u+[u]$, $\mathbf v \gets \mathbf v+[v]$
                    \State $W \gets W + Fu + Fv$ \smallskip
                }
            }
        }
        \State \Return $\mathbf u, \mathbf v$
    \end{algorithmic}
\end{algorithm}

\begin{maproposition}\label{algoHERMIT} 
Algorithm~\ref{directsumhermitalgo} is correct.
\end{maproposition}

\begin{proof}
It follows from the construction that, after the first successful iteration of the loop, $(u,v)$ is a hyperbolic pair in $V$.
Indeed, we notice that the subspace $H_1$ generated by $u$ and $v$ does not change throughout the loop, and so it is still a plane at the end. Moreover, each update successively ensures that $\langle u,u \rangle = 0$, then $\langle v,v \rangle = 0$ and finally $\langle u,v \rangle = 1$.
We observe that $\langle v,u \rangle$ does not vanish on line~10 because the substitution of line~9 leaves it unchanged.
After this, we update $W$ so that we continue to work in the orthogonal complement of $H_1$ which have dimension $2(s{-}1)$ and Witt index $s{-}1$ thanks to Witt's cancellation theorem. The induction then goes.
\end{proof}

\begin{monlemme}
\label{lem:random:proba}
The tests of lines~4 and~7 are successful if and only if the vectors $u, v$ picked on line~3 span a hyperbolic plane of $\Ortho{W}$ and $v$ is not isotropic.

Moreover, this happens with probability at least $\frac{\sqrt{q_F}-1}{\sqrt{q_F}+1} \cdot \frac{q_F - 1}{q_F} \geq \frac 4 9$.
\end{monlemme}

\begin{proof}
It is clear that if $u$ and $v$ pass all tests, then they span a hyperbolic plane and that $v$ is nonisotropic.
Conversely, we need to prove if $H \subset \Ortho{W}$ is a hyperbolic plane and $u, v$ span $H$ with $\langle v, v \rangle \neq 0$, then all tests pass.
It is obvious for the test of line~4. If, on line~7, we had $\langle u, v \rangle = 0$, then $u$ would be orthogonal to both $u$ and $v$, implying that the Hermitian form would be degenerated on $H$. This is a contradiction.

We now count to number of hyperbolic planes in $\Ortho{W}$. For this, we write $\dim_F \Ortho{W} = 2d$ and we consider the map
$$\begin{array}{rcl}
\mathcal S : \quad 
\Big\{\,\text{\begin{tabular}{@{}c@{}}pair of nonorthogonal\\[-1.5ex]isotropic vectors in $\Ortho{W}$\end{tabular}}\,\Big\} 
  & \longrightarrow & \big\{\,\text{hyperbolic planes $\subset \Ortho{W}$}\,\big\} \\
(x,y) & \mapsto & Fx + Fy.
\end{array}$$
The cardinality of the domain of $\mathcal S$ is $\IsoVect(\Ortho{W}) \cdot \IsoVects(\Ortho{W})$ 
where $\IsoVect(\Ortho{W})$ denotes the number of isotropic vectors in $\Ortho{W}$ (as in the proof of Theorem~\ref{SEGRE}) and $\IsoVects(\Ortho{W})$ is the number of isotropic vectors in $\Ortho{W}$ which are nonorthogonal to a given one.
By the proof of Theorem~\ref{SEGRE}, we know that $\IsoVect(\Ortho{W}) = (q_F^d-1)(q_F^{d-1/2}+1)$.
The computation of $\IsoVects(\Ortho{W})$ can be handled similarly: reusing the notation of the proof of Theorem~\ref{SEGRE}, it corresponds to count the tuples
$((a_i)_{1 \leq i \leq d}, (b_i)_{1 \leq i \leq d}, \alpha)$ with $\sum_{i=1}^d a_i \sigma(b_i) = \alpha = -\sigma(\alpha)$ and $b_1 \neq 0$.
There are then $\sqrt{q_F}$ choices for $\alpha$, $(q_F - 1) q_F^{d-1}$ choices for the $b_i$ and when all those are fixed, there remains $q_F^{d-1}$ choices for the $a_i$.
Therefore $\IsoVects(\Ortho{W}) = (q_F - 1) q_F^{2d-\frac 3 2}$ and we finally obtain:
$$\text{Card}(\text{domain of } \mathcal S) = (q_F^d-1) \cdot (q_F - 1) \cdot (q_F^{3d-2} + q_F^{2d-\frac 3 2}).$$
Similarly, given a hyperbolic plane $H \subset \Ortho{W}$, we have:
$$\text{Card}(\mathcal S^{-1}(H)) = \IsoVect(H) \cdot \IsoVects(H) = (q_F-1)^2 \cdot (q_F + \sqrt{q_F}).$$
We conclude that the number of hyperbolic planes in $\Ortho{W}$ is
$$A = \frac{(q_F^d-1) \cdot (q_F^{3d-2} + q_F^{2d-\frac 3 2})}{(q_F-1) \cdot (q_F + \sqrt{q_F})}.$$
Now, once a hyperbolic place $H$ is fixed, the number of possibilities for $v$ is
$$B = (q_F^2 - 1) - (1 + \sqrt{q_F})(q_F-1) = (q_F-1)(q_F - \sqrt{q_F})$$
while the number of options for $u$ is $C = q_F^2 - q_F$.
Finally, the probability we are looking for is $\frac{A B C}{q_F^{4d}}$ and
calculus shows that it is always greater than $\frac{\sqrt{q_F}-1}{\sqrt{q_F}+1} \cdot \frac{q_F - 1}{q_F}$ (which is the limit when $d$ goes to infinity).
The fact that the latter is bounded from below by $\frac 4 9$ follows from the observation that $q_F$ is necessarily at least $9$ because $F$ has odd characteristic and admits a subfield of index~$2$.
\end{proof}

\begin{maproposition}
\label{termHERMIT}
Algorithm~\ref{directsumhermitalgo} terminates almost surely and its average complexity is $O(r^3)$ operations in $F^\sigma$
and $O(r)$ computations of inner products.
\end{maproposition}

\begin{proof}
Termination follows directly for Lemma~\ref{lem:random:proba}.
Regarding complexity, we claim that each successful iteration of the loop costs at most $O(r^2)$ operations in $F^\sigma$ and $O(1)$ computations of inner products.
To achieve this, we first observe that solving the equation~$(\mathcal{E}_{u,v})$ amounts to computing the attached discriminant (which corresponds to computing $4$ inner products) and finding a uniformly distributed preimage of the discriminant by the norm map; the latter can be done using the algorithms of~\cite{18} for a constant cost.
Similarly solving~$(\mathcal{E}_{v,u})$ reduces to a linear system, which can be attacked by simple linear algebra over $F^\sigma$ for a constant cost again.
Regarding the computation of $W$, we may process as follows: we maintain a matrix $M$ in reduced row echelon form representing the subspace of $V^\star = \text{Hom}_F(V, F)$ generated by the forms $\langle -, w \rangle$ with $w \in W$.
At each update of $W$ on line~11, we need to add two new lines to $M$ and re-echelon it; this has a cost of $O(r^2)$ operations in $F$ using standard Gaussian elimination.
Moreover, knowing $M$, sampling $u$ and $v$ on line~3 amounts to finding two random solutions of the linear system $M X = 0$. Since $M$ is already row-echeloned, this can be done for a cost of $O(r^2)$ operations in $F$ as well.
\end{proof}

Finally, the link between Algorithm~\ref{directsumhermitalgo} and the question we are interested in is established in the next proposition.

\begin{maproposition}
\label{prop:outputHERMIT}
If $\mathbf u, \mathbf v$ is the output of Algorithm~\ref{directsumhermitalgo}, then the space generated by $\mathbf u$ is a uniformly distributed random isotropic subspace of $V$ of dimension $s$.
\end{maproposition}

\begin{proof}
The fact that the span of $\mathbf u$ is an isotropic subspace of dimension $s$ is clear.
To prove that it is uniformly distributed, we notice that, after line~3, the plane $H := Fu + Fv$ is uniformly distributed among all planes in $\Ortho{W}$.
Since this plane stays unchained throughout the loop, the first part of Lemma~\ref{lem:random:proba} implies that, at the end of the loop, $H$ is uniformly distributed among all hyperbolic planes in $\Ortho{W}$.

We now fix a hyperbolic plane $H \subset \Ortho{W}$, together with a nonisotropic vector $v \in H$.
We claim that, when $u$ varies in $H$, the vector $u$ one obtains after the replacement of line~6 is uniformly distributed in the set $\mathcal I_H$ of isotropic vectors in $H$.
In order to prove this, for $u \in H$, we define $L(u) \subset H$ as the affine line passing through $u$ and directed by $v$.
We also set $S(u) := L(u) \cap \mathcal I_H$. Clearly, for any fixed $u$ noncollinear to $v$, the $L(\alpha u)$ form a partition of $H \backslash Fv$ when $\alpha$ varies in $F^\times$.
Since $v$ is itself nonisotropic, we conclude that
\begin{equation}
\label{eq:IH}
\mathcal I_H = \bigsqcup_{\alpha \in F^\times} S(\alpha u).
\end{equation}
Moreover, the multiplication by $\alpha$ defines a bijection $S(u) \to S(\alpha u)$; hence, all the $S(\alpha u)$ have the same cardinality.
Coming back now to the algorithm, we notice that the effect of lines~5 and~6 is to replace $u$ by a uniformly distributed random vector in $S(u)$.
The decomposition~\eqref{eq:IH}, combined with the fact that all $S(\alpha u)$ have the same cardinality, then implies that the vector $u$ obtained after line~6 gets uniformly distributed in $\mathcal I_H$ when $u$ varies on any given line of $H$ that does not contain $v$.
Since this holds for any line, our claim is proved.

Let $\mathcal A$ be the set of all $\big((H_1, u_1), \ldots, (H_s, u_s)\big)$ such that the $H_i$ are pairwise orthogonal hyperbolic planes in $V$ and, for all $i$, $u_i$ is an isotropic vector in $H_i$.
It follows for what precedes that, if $\mathbf u = (u_1, \ldots, v_s)$, $\mathbf v = (v_1, \ldots, v_s)$ is the output of Algorithm~\ref{directsumhermitalgo}, the tuple $\big((H_1, u_1), \ldots, (H_s, u_s)\big)$ with $H_i := F u_i + F v_i$ is uniformly distributed in $\mathcal A$.
To conclude, it is then enough to prove that all the fibers of the map
$$\begin{array}{rcl}
\mathcal A & \longrightarrow & \mathcal B := \Big\{ \text{\begin{tabular}{@{}c@{}} $s$-dimensional isotropic\\[-1.5ex]subspaces of $V$\end{tabular}} \Big\} \\
\big((H_1, u_1), \ldots, (H_s, u_s)\big) & \mapsto & F u_1 + \cdots + F u_s
\end{array}$$
have the same cardinality.
For this, we use the fact that the unitary group $\text{U}(V)$ acts transitively of $\mathcal B$.
In other words, given two $s$-dimensional isotropic subspaces $U, U' \subset V$, there exists a unitary transformation $f : V \to V$ such that $f(U) = U'$.
Such an $f$ induces a bijection between the fibers above $U$ and $U'$, then proving that the cardinalities are equal.
\end{proof}

\begin{maremarque}
Two important optimizations of Algorithm~\ref{directsumhermitalgo} needs to be mentioned.
Firstly, the lines 8--10 are in fact not needed if we are only interested in the span of $\mathbf u$ (as we are here).
We kept them because computing a full hyperbolic basis can be interested on its own, and will be actually used afterwards when we will come to enumeration.
Secondly, if the picked vector $v$ is not isotropic on line~4, instead of giving up immediately, we may try to swap $u$ and $v$ and redo the test.
On the one hand, this increases the probability of success and, on the other hand, one can prove that this does not affect the fact that the span of $\mathbf u$ is uniformely distributed at the end.
\end{maremarque}

\subsubsection{The Euclidean case}
\label{sssec:random:euclidean}

We move to the Euclidean case, \emph{i.e.} we consider the same setting as before expect that we now assume that $\sigma$ is the identity.
The equation $(\mathcal E_{u,v})$ continues to make sense but its resolution is a bit different. Precisely, expanding the scalar product, we find that it is equivalent to
$$\langle u, u \rangle + 2 \lambda \cdot \langle u,v \rangle + \lambda^2 \cdot \langle v, v \rangle = 0.$$
If $\langle v, v \rangle = 0$, it is a linear equation that we can solve as soon as $\langle u,v \rangle \neq 0$.
On the contrary, if $\langle v,v \rangle \neq 0$, it is a quadratic equation whose 
(reduced) discriminant is $\Delta := \langle u,v \rangle^2 - \langle u, u \rangle \cdot \langle v, v \rangle$ (this is in fact the same as before!).
The equation $(\mathcal E_{u,v})$ has no solution when $\Delta$ is not a square and it has one or two solutions otherwise: if $\delta^2 = \Delta$, they are given by
$\lambda = \frac{\delta - \langle u,v \rangle}{\langle v,v \rangle}$.

From here, we can write down Algorithm~\ref{directsumeuclidalgo} (which is a direct translation of Algorithm~\ref{directsumhermitalgo}).

\begin{algorithm}
    \caption{Direct sum decomposition of hyperbolic planes in the Euclidean case}\label{directsumeuclidalgo}
    \begin{algorithmic}[1]
        \Require{$V$: the ambient Euclidean vector space}
        \Ensure{$\mathbf u, \mathbf v$: a basis of hyperbolic pairs}
        \State $\mathbf u,\mathbf v,W \gets [\:],[\:], 0$
        \State \WhileCustom{$W \neq V$}{
        \State $u, v \gets $ two random vectors in $\Ortho{W}$
        \State \IfCustom{\rm $(u,v)$ are linearly independent and $\langle v,v \rangle \neq 0$}{
        \State $\Delta \gets \langle u,v \rangle^2 - \langle u, u \rangle \cdot \langle v, v \rangle$
        \State \IfCustom{\rm $\Delta$ is a square in $F$}{
                \State $\lambda \gets $ a random solution of the equation $(\mathcal{E}_{u,v})$
                \State $u \gets u + \lambda v$ \hfilll \Comment{now $\langle u,u \rangle = 0$}
                \State \IfCustom{$\langle u, v \rangle \neq 0$}{
                    \State $\lambda \gets $ a solution of the equation $(\mathcal{E}_{v,u})$
                    \State $v \gets v + \lambda u$ \hfilll \Comment{now $\langle v,v \rangle = 0$}
                    \State $v \gets v / \langle v, u \rangle$ \hfilll \Comment{now $(u,v)$ is a hyperbolic pair}
                    \State $\mathbf u \gets \mathbf u+[u]$, $\mathbf v \gets \mathbf v+[v]$
                    \State $W \gets W + Fu + Fv$ \smallskip
        }}}}
        \State \Return $\mathbf u, \mathbf v$
    \end{algorithmic}
\end{algorithm}

\begin{maproposition}\label{algook} 
Algorithm~\ref{directsumeuclidalgo} is correct.
It terminates almost surely and its average complexity is $O(r^3)$ operations in $F$
and $O(r)$ computations of inner products.
Moreover, if $\mathbf u, \mathbf v$ is the output of Algorithm~\ref{directsumeuclidalgo}, then the space generated by $\mathbf u$ is a uniformly distributed random isotropic subspace of $V$ of dimension $s$.
\end{maproposition}

\begin{proof}
It is a repetition of the proofs of Propositions~\ref{algoHERMIT}, \ref{termHERMIT} and~\ref{prop:outputHERMIT}
(with the small difference that the probability of success in the analogue of Lemma~\ref{lem:random:proba} is now bounded from below by $\frac{(q_F-1)^2}{2q_F^2} \geq \frac 2 9$).
\end{proof}

\subsection{Enumeration of selfdual skew cyclic codes}

We finally address the question of enumeration.
As we already said earlier, an algorithm that outputs in one shot the complete list of selfdual codes in $\mathbf E_k$ would only have a limited interest because the number of such codes grows exponentially with respect to~$r$.

Instead, we will work with iterators, that are, roughly speaking, procedures that produce a new item each time they are called, without precomputing the entire list at the beginning.
Concretely, we model iterators by importing the keywords \yield and \nextc from the Python syntax.
When a procedure containing the keyword \yield is called, it is not executed but instead returns an object called \emph{iterator}, which can be understood as a pointer to the current state of execution of the procedure.
Now, each time the iterator is called through the keyword \nextc, the execution of the procedure continues until a statement \yield is encoutered; at that point, the execution is interrupted and the iterator outputs the attribute of the \yield instruction.

In all what follows, we assume\footnote{Such an iterator is available in many softwares, including SageMath. It is moreover easy to construct: we iterate over the subset of $I \subset \{1, \ldots, n\}$ of cardinality $m$ and, for each such $I$, we run over all the matrices in reduced row echelon form with pivots at positions in $I$.} that we have at our disposal, for all integers $m \leq n$ and any finite field $F$, an iterator producing the list of all  matrices in reduced row echelon form with $m$ rows and $n$ columns.
We note that such matrices are in one-to-one correspondence with $m$-dimensional $F$-linear subspaces of $F^n$ (the subspace being the span of the rows of the matrix).
In a similar fashion, we also assume that, for any given linear system, we have at our disposal an iterator running over its solutions.

We now explain how to build iterators over each component of the product $W_{\mathtt{nonpal}} \times W_{\mathtt{pal}}$.

\subsubsection{The nonpalindromic case}

In this case, we have to construct an iterator running over all $\mathbf F_l$-linear subspaces of $\mathbf K_l$.
In order to reduce this task to a matrix enumeration, we first pick a basis of $\mathbf K_l$ over $\mathbf F_l$ (this can be done easily; for example, a basis of $\mathbf K$ over $\mathbf F$ does the job).
Once this is achieved, we take an iterator that runs over all  matrices over $\mathbf F_l$ in reduced row echelon form with $r$ columns, which directly solves our problem.

\subsubsection{The Euclidean case}
\label{sssec:enum:euclidean}

As in Subsection~\ref{sssec:random:euclidean}, we work with a general $r$-dimensional Euclidean space $V$ over a finite field $F$ of cardinality $q_F$ and assume that $V$ has Witt index~$s$.
By the results of Subsection~\ref{sssec:random:euclidean}, we can further assume that we are given a hyperbolic basis of $V$, that is a basis $(u_1, \ldots, u_s, v_1, \ldots, v_s)$ such that $\langle u_i, v_i \rangle = 1$ and all other scalar products between elements in the basis vanish.

In order to take advantage of this basis, we will enumerate the $s$-dimensional subvector spaces of $V$ in a slightly different manner.
Those spaces are parametrized by the  matrices $M$ in reduced row echelon form of size $s \times (2s)$, but we shall further split $M$ and write it as a block matrix as follows:
$$M = \left( \begin{matrix} A & B \\ 0 & C \end{matrix} \right).$$
Here $A$, $B$ and $C$ all have $s$ columns and the horizontal separation is positionned is such a way that the last line of $A$ is not identically zero.
The matrices $A$ and $C$ are then  reduced row echelon matrices of size $d \times s$ and $(s{-}d) \times s$ respectively (for some $d$).
Besides, the columns of $B$ in front of the pivots of $C$ all vanish.
Conversely, if we choose $A$, $B$ and $C$ satisfying the above conditions, the resulting block matrix $M$ will be in reduced row echelon form.
In other words, there is a bijection between the matrices $M$, on the one hand, and the triples $(A, B, C)$, on the other hand; in the sequel, we will constantly rely on it to enumerate the $M$.

\begin{maremarque}
At the level of cardinalities, the above bijection leads to the (classical) formula
$$\qbin{2s}{s}{q_F}=\sum_{d=0}^{s} q_F^{d^2} \cdot \qbin{s}{d}{q_F}^{2}.$$
\end{maremarque}

The $(A,B,C)$-presentation is quite interesting for our purpose because the isotropy condition translates to the equations:
\begin{align}
A B^{\mathsf{tr}} + B A^{\mathsf{tr}} & = 0 \label{eq:onB} \\
A C^{\mathsf{tr}} & = 0 \label{eq:onC}
\end{align}
Equation~\eqref{eq:onC} means that the row-span of $A$ should be orthogonal to the row-span of $C$ for the standard scalar product on $F^s$.
Since those two spaces have completary dimension, we conclude that $\text{RowSpan}(C)$ must be the orthogonal of $\text{RowSpan}(A)$.
Given that, in addition, $C$ must also be in reduced row echelon form, we conclude that $C$ is uniquely determined by $A$: it is the reduced row echelon basis of $\Ortho{\text{RowSpan}(A)}$.

Once $C$ is known, one also knows its pivots and the shape of $B$ is determined.
Equation~\eqref{eq:onB} then appears as a linear equation on the entries of $B$, which can be easily solved using Gaussian elimination.

All of this leads to Algorithm~\ref{algo:enum:euclidean}.

\begin{algorithm}
    \caption{Iterator over maximal isotropic spaces (Euclidean case)}\label{algo:enum:euclidean}
    \begin{algorithmic}[1]
        \State $\mathcal A \gets $ iterator over reduced row echelon matrices over $F$ with $s$ columns
        \State \WhileCustom{$A \gets \nextc(\mathcal A)$}{
          \State $C \gets $ reduced row echelon basis of $\Ortho{\text{RowSpan}(A)}$
          \State $\mathcal B \gets $ iterator over solutions of \eqref{eq:onB} with vanishing columns in front of pivots of $C$
          \State \WhileCustom{$B \gets \nextc(\mathcal B)$}{
            \State $\yield \left( \begin{matrix} A & B \\ 0 & C \end{matrix} \right)$
          }
        }
    \end{algorithmic}
\end{algorithm}

Regarding complexity, it is clear that, in the worst case, an iteration of Algorithm~\ref{algo:enum:euclidean} requires at most $O(r^6)$ operations in $F$ since it only involves Gaussian elimination in dimension at most $O(r^2)$.
However, in most cases, an iteration only consists in going from one solution $B$ to the next one; once a basis of the space of solutions has been computed, this costs only $O(r^2)$ operations in $F$.

\begin{maremarque}
Denoting by $d$ the number of rows of $A$, one can prove that the linear system~\eqref{eq:onB} consists of $\frac{d(d+1)} 2$ linearly independent equations.
Therefore, the set of admissible $B$ is a $F$-vector space of dimension $\frac{d(d-1)} 2 = \binom d 2$; hence it has cardinality $q_F^{\binom d 2}$.
From this, we derive that the number of isotropic subspaces of $V$ of dimension $s$ is equal to
$$\sum_{d=0}^s q_F^{{d\choose 2}}\qbin{s}{d}{q_F}.$$
Comparing with Segre's formula (see Theorem~\ref{SEGRE}), we find the identity
$$\prod_{d=0}^{s-1} (1+q_F^d) = \sum_{d=0}^s q_F^{{d\choose 2}}\qbin{s}{d}{q_F}$$
which is actually a special case of the well-known polynomial identity~\cite{19}:
\begin{equation}
\label{eq:qbin}
\prod_{k=0}^{n-1} (1+q^k t) = \sum_{k=0}^n q^{\binom k 2}\qbin{n}{k}{q} t^k.
\end{equation}
As a byproduct, our approach then provides a \emph{bijective proof} of this identity when $t = 1$ and $q$ is a power of a prime number.
\end{maremarque}

\subsubsection{The Hermitian case}

We now equip $F$ with a nontrivial involution $\sigma : F \to F$ and assume that the pairing $\langle -, - \rangle$ on $V$ is $\sigma$-sesquilinear.
In this new situation, all the discussion of Subsection~\ref{sssec:enum:euclidean} applies, except that the Equations~\eqref{eq:onB} and~\eqref{eq:onC} have to be replaced by the following ones:
\begin{align}
A \sigma(B^{\mathsf{tr}}) + B \sigma(A^{\mathsf{tr}}) & = 0 \label{eq:onBherm} \\
A \sigma(C^{\mathsf{tr}}) & = 0 \label{eq:onCherm}
\end{align}
As in the Euclidean case, it turns out that Equation~\eqref{eq:onCherm} fully determines $C$; precisely $C$ is the reduced row echelon basis of $\Ortho{\text{RowSpan}(\sigma(A))}$.
Similarly, Equation~\eqref{eq:onBherm} provides a linear system on the entries on $B$ but we need to careful that it is $F^\sigma$-linearity and not $F$-linearity as before.
Anyway, the system can equally be solved using Gaussian elimination.

Taking these remarks into account, we end up with Algorithm~\ref{algo:enum:hermitian}

\begin{algorithm}
    \caption{Iterator over maximal isotropic spaces (Hermitian case)}\label{algo:enum:hermitian}
    \begin{algorithmic}[1]
        \State $\mathcal A \gets $ iterator over reduced row echelon matrices over $F$ with $s$ columns
        \State \WhileCustom{$A \gets \nextc(\mathcal A)$}{
          \State $C \gets $ reduced row echelon basis of $\Ortho{\text{RowSpan}(\sigma(A))}$
          \State $\mathcal B \gets $ iterator over solutions of \eqref{eq:onBherm} with vanishing columns in front of pivots of $C$
          \State \WhileCustom{$B \gets \nextc(\mathcal B)$}{
            \State $\yield \left( \begin{matrix} A & B \\ 0 & C \end{matrix} \right)$
          }
        }
    \end{algorithmic}
\end{algorithm}

\begin{maremarque}
Similarly to the Euclidean case, our approach gives a bijective proof of the numerical identity
$$\prod_{d=0}^{s-1} (1+q_F^{d + 1/2}) = \sum_{d=0}^s q_F^{d^2/2}\qbin{s}{d}{q_F}$$
which is Equation~\eqref{eq:qbin} evaluated at $q = q_F$ and $t = \sqrt q$.
\end{maremarque}

\subsection{An implementation in SageMath}\label{section4}

We implemented the algorithms of this section in SageMath. Our package is available at

\hfill\url{https://plmlab.math.cnrs.fr/caruso/selfdual-skew-cyclic-codes}\hfill\null

It consists in a main class instantiated with the extension $\mathbf{K/F}$ of order $r$ and a palindromic polynomial of the center $P(X^r)$ in $\mathbf{F}(X^{\pm r})$ of $\mathbf{K}[X^{\pm 1};\Frob]$ as constructor parameters.
It provides an iterator on all selfdual codes for the Ore algebra $\mathbf{K}[X^{\pm 1};\Frob]/P(X^r)$.
Hereunder, we present a bunch of examples covering all the encoutered situations : palindromic Euclidean and palindromic Hermitian.

We start by loading our package and defining the relevant base rings.

    \begin{sagecommandline}
        sage: load("selforthogonal_codes.sage")
        sage: q = s = 3; F = GF(q); Fy.<y> = F[]
        sage: Q = F['z'].irreducible_element(2*s, "adleman-lenstra")
        sage: Q
    \end{sagecommandline}

\begin{moncas}{Palindromic Euclidean: $q=3$, $s=3$ and $P(Y)=Y-1$}
    \begin{sagecommandline}
        sage: A = SelfDualCodes(y - 1, Q) 
        sage: iter = A.enumerate_selfdual_codes()
        sage: next(iter)
        sage: next(iter)
    \end{sagecommandline}
\end{moncas}
\begin{moncas}{Palindromic Hermitian case: $q=3$, $s=3$ and $P(Y)=Y^2+1$}
    \begin{sagecommandline}
        sage: A = SelfDualCodes(y^2 + 1, Q) 
        sage: iter = A.enumerate_selfdual_codes()
        sage: next(iter)
    \end{sagecommandline}
\end{moncas}

Benchmarks for a larger set of inputs are reported on Figures~\ref{fig:s2}, \ref{fig:s3} and~\ref{fig:s4}; there were run on computer with Intel(R) Core(TM) i7-9750H CPU 2.60GHz processor x64 and 16 GB of RAM.

\begin{figure}
        \hfill\begin{tabular}{ |c|c|c|c|c| }
            \hline
            $P(Y)$  & $q=3$               & $q=5$               & $q=7$       & $q=3^2$     \\
            \hline
            $Y-1$   & no codes.           & no codes.           & no codes.   & no codes.   \\[0.5ex]
            \hline
            $Y^3-1$ & inseparable         & no codes.           & no codes.   & inseparable \\[0.5ex]
            \hline
            $Y^5-1$ & no codes.           & inseparable         & no codes.   & no codes.   \\[0.5ex]
            \hline
            $Y^7-1$ & no codes.           & no codes.           & inseparable & no codes.   \\[0.5ex]
            \hline
            $Y^9-1$ & inseparable         & no codes.           & no codes.   & inseparable \\[0.5ex]
            \hline
            $Y+1$   & \thead{ 9 ms                                                      }      & \thead{ 9 ms                                          }       & \thead{ 16 ms                                          }        &  \thead{ 21 ms                                          } \\[0.5ex]
            \hline
            $Y^2+1$ & \thead{ 16 ms                                                   }         &\thead{ 6 ms                                           }     & \thead{ 15 ms                                           }       & \thead{ 15 ms                                          }      \\[0.5ex]
            \hline
            $Y^3+1$ & inseparable         & \thead{ 26 ms                              }         & \thead{ 22 ms                                            }      & inseparable \\[0.5ex]
            \hline
            $Y^4+1$ & \thead{ 18 ms                                                    }        & \thead{ 21 ms                                             }       & \thead{ 35 ms                                             }       & \thead{ 48 ms                                            }      \\[0.5ex]
            \hline
            $Y^5+1$ & \thead{ 62 ms                                                    }          & inseparable & \thead{ 111 ms                                           }       & \thead{ 128 ms                                            }        \\[0.5ex]
            \hline
            $Y^6+1$ & inseparable         & \thead{ 47 ms                              }       & \thead{ 59 ms                                          } & inseparable \\[0.5ex]
            \hline
            $Y^7+1$ & \thead{ 80 ms                                                    }        & \thead{ 300 ms                                             }      & inseparable & \thead{ 250 ms                                             }     \\[0.5ex]
            \hline
            $Y^8+1$ & \thead{ 463 ms                                                  }               & \thead{ 87 ms                            }      & \thead{ 113 ms                                            . }      & \thead{ 108 ms                                         }     \\[0.5ex]
            \hline
            $Y^9+1$ & inseparable         & \thead{ 218 ms                              }       & \thead{ 125 ms                                              }      & inseparable \\[0.5ex]
            \hline
        \end{tabular}\hfill\null
\caption{Timings for the generation of one single code when $s=2$}\label{fig:s2}
\end{figure}

\begin{figure}
        \hfill\begin{tabular}{ |c|c|c|c|c| }
            \hline
            $P(Y)$  & $q=3$               & $q=5$      & $q=7$               & $q=3^2$           \\
            \hline
            $Y-1$   & \thead{ 21 ms                       }       & no codes. & \thead{ 207 ms                     }      & no codes.   \\[0.5ex]
            \hline
            $Y^3-1$ & inseparable         & no codes.      & \thead{ 42 ms                        }       & inseparable      \\[0.5ex]
            \hline
            $Y^5-1$ & \thead{ 101 ms                        }     & inseparable   & \thead{ 129 ms                     }     & no codes.  \\[0.5ex]
            \hline
            $Y^7-1$ & \thead{ 195 ms                      }      & no codes.  & inseparable         & no codes.      \\[0.5ex]
            \hline
            $Y^9-1$ & inseparable         & no codes.         & \thead{ 342 ms                      }      & inseparable  \\[0.5ex]
            \hline
            $Y+1$   & no codes.           & \thead{ 21 ms  }      & no codes.           & \thead{ 56 ms   }     \\[0.5ex]
            \hline
            $Y^2+1$ & \thead{ 152 ms                       }      &  \thead{ 12 ms                              }     & \thead{ 36 ms                         }       &  \thead{ 32 ms                              }       \\[0.5ex]
            \hline
            $Y^3+1$ & inseparable         & \thead{ 57 ms   }  & no codes.           & inseparable        \\[0.5ex]
            \hline
            $Y^4+1$ & \thead{ 38 ms                        }      &  \thead{ 47 ms                             }   & \thead{ 74 ms                      }       &  \thead{ 141 ms                               }        \\[0.5ex]
            \hline
            $Y^5+1$ & no codes.           & inseparable     & no codes.           & \thead{ 317 ms  }      \\[0.5ex]
            \hline
            $Y^6+1$ & inseparable         & \thead{ 101 ms  }   & \thead{ 139 ms                       }      & inseparable     \\[0.5ex]
            \hline
            $Y^7+1$ & no codes.           & \thead{ 398 ms  }     & inseparable         & \thead{ 601 ms . }        \\[0.5ex]
            \hline
            $Y^8+1$ & \thead{ 209 ms                        }   &  \thead{ 270 ms                              }   & \thead{ 270 ms                       }     &  \thead{ 280 ms                               }      \\[0.5ex]
            \hline
            $Y^9+1$ & inseparable         & \thead{ 450 ms  }    & no codes.           & inseparable   \\[0.5ex]
            \hline
        \end{tabular}\hfill\null
\caption{Timings for the generation of one single code when $s=3$}\label{fig:s3}
\end{figure}

\begin{figure}
            \hfill\begin{tabular}{ |c|c|c|c|c| }
                \hline
                $P(Y)$  & $q=3$               & $q=5$   & $q=7$    & $q=3^2$                   \\
                \hline
                $Y-1$   & no codes.           & no codes.  & no codes.    & no codes.            \\[0.5ex]
                \hline
                $Y^3-1$ & inseparable         & no codes.     & no codes.    & inseparable            \\[0.5ex]
                \hline
                $Y^5-1$ & no codes.           & inseparable   & no codes.    & no codes.         \\[0.5ex]
                \hline
                $Y^7-1$ & no codes.           & no codes.     & inseparable   & no codes.            \\[0.5ex]
                \hline
                $Y^9-1$ & inseparable         & no codes.      & no codes.    & inseparable          \\[0.5ex]
                \hline
                $Y+1$   & \thead{ 59 ms                       }       &  \thead{ 49 ms                              }   & \thead{ 58 ms  }      & \thead{ 177 ms }       \\[0.5ex]
                \hline
                $Y^2+1$ & \thead{ 78 ms                       }       &  \thead{ 29 ms                            }      & \thead{ 89 ms  }   & \thead{ 69 ms   }    \\[0.5ex]
                \hline
                $Y^3+1$ & inseparable         & \thead{ 128 ms  }        & \thead{ 90 ms   }   & inseparable    \\[0.5ex]
                \hline
                $Y^4+1$ & \thead{ 88 ms                         }      &  \thead{ 108 ms                            }    & \thead{ 174 ms}   & \thead{ 412 ms   }         \\[0.5ex]
                \hline
                $Y^5+1$ & \thead{ 220 ms                       }      & inseparable   & \thead{ 336 ms}  & \thead{ 723 ms   } \\[0.5ex]
                \hline
                $Y^6+1$ & inseparable         & \thead{ 200 ms  }  & \thead{ 388 ms }   & inseparable    \\[0.5ex]
                \hline
                $Y^7+1$ & \thead{ 286 ms                    }      &  \thead{ 387 ms                            }   & inseparable    & \thead{ 1367 ms  }     \\[0.5ex]
                \hline
                $Y^8+1$ & \thead{ 406 ms                      . }     &  \thead{ 551 ms                              }   & \thead{ 586 ms }   & \thead{ 2159 ms  }     \\[0.5ex]
                \hline
                $Y^9+1$ & inseparable         & \thead{ 691 ms  } & \thead{ 784 ms}  & inseparable     \\[0.5ex]
                \hline
            \end{tabular}\hfill\null
\caption{Timings for the generation of one single code when $s=4$}\label{fig:s4}
\end{figure}
\section{Enumeration of purely inseparable selfdual skew cyclic codes}\label{inseparableCase}
We now address the case where $k$ is not coprime to the characteristic $p$. We aim at finding an enumeration algorithm of selfdual skew cyclic codes in this case too.
If $k$ decomposes as $k'p^m$ with $k'$ coprime with $p$, it follows easily from the chinese remainder isomorphism 
\begin{align*}
\mathbf{E}_{k} 
& \simeq \mathbf{K}[Y,X;\Frob]/(Y^{k'}-1,X^{rp^m}-Y) \\
& \simeq (\mathbf{K}[Y,X;\Frob]/(Y^{k'}-1))/(X^{r}-Y^{\frac{1}{p^m}})^{p^m} \\
& \simeq \left(\mathbf{K}[Y,X;\Frob]/\prod\limits_{1 \leq l \leq n} P_l(Y) \right)/(X^{r}-Y^{\frac{1}{p^m}})^{p^m} \\
& \simeq \prod\limits_{1 \leq l \leq n} \left(\mathbf{K}[Y,X;\Frob]/P_l(Y)\right)/(X^{r}-Y^{\frac{1}{p^m}})^{p^m}
\end{align*}
that we can recover an enumeration algorithm for any $k$ by combining the separable case and the case where $k = p^m$ (purely inseparable case).

\subsection{Enumeration of purely inseparable selfdual skew cyclic codes}

In order to solve the purely inseparable case, we follow a factorization approach, inspired by but slightly different from that of article~\cite{4}.
We introduce twisted skew separable codes $\mathbf{E}^{(\xi X^{t})}_{k,l}$, that are slight generalizations of previously considered skew separable codes.
They are defined as skew separable codes of $\mathbf{E}_{k}^{(l)}$ corresponding to the usual adjunction on $\mathbf{E}_{k}^{(l)}$ composed with the conjugation by $\xi X^{t}$ for $\xi \in \mathbf{K}_l$  and $t \in \{0,s\}$.
We will then obtain all inseparable selfdual codes as products of twisted skew separable selfdual codes.

\begin{madefinition}
    We fix parameters $t \in \{0,s\}$, $\xi \in \mathbf{K}_l$. We denote by $\mathbf{E}^{(\xi X^{t})}_{k,l}$, the space $\mathbf{E}_{k}^{(l)}$ equipped with the $\xi X^{t}$-twisted bilinear form $(\kappa,\rho)=\mathrm{Trace}_{\mathbf{K}_l/\mathbf{F}_l}(\zeta.\xi \kappa \Frob^{t}(\sigma_l(\rho)))$.
\end{madefinition}
The corresponding \textit{adjunction} is $f^{\bullet_{X^{t} \xi^{-1}}}=X^{t} {\xi^{-1}\zeta^{-1}} \sum_i X^{-i} \sigma_l(f_i) {\zeta\xi}X^{-t} $; we have
$$(\kappa,f(\rho))_{{\mathbf{F}_l}}^{(\xi X^{t})} = (f^{\bullet_{X^{t} \xi^{-1}}}(\kappa),\rho)_{{\mathbf{F}_l}}^{(\xi X^{t})}.$$
In the sequel, we will take $\sigma_l(\xi)=\xi$, and if $t=s$ $\Frob^t(\xi) = -\xi $, so that the $\xi X^{t}$-twisted bilinear form enjoys following symmetries:
\begin{itemize}
\item it is Euclidean if $y_l=\pm1$ and $t=0$,
\item it is Hermitian if $y_l\neq\pm1$ and $t=0$,
\item it is skew-Euclidean if $y_l=\pm1$ and $t=s$,
\item it is skew-Hermitian if $y_l\neq \pm1$ and $t=s$.
\end{itemize}

\begin{maremarque}
\label{rem:segre2}
Reusing the method of Remark~\ref{rem:segre}, we can compute the number of twisted codes when $k = 1$.
For example, in the skew-Euclidean case, it is given by
$$\frac{\text{Card}\big(\text{Sp}_{2s}(\mathbf F_l)\big)}{q_l^{s(s+1)/2} \cdot \text{Card}\big(\GL s {\mathbf F_l}\big)}
= \prod_{d=1}^s (1 + q_l^d)$$
where $\text{Sp}_{2s}$ stands for the symplectic group. We refer to~\cite{9} for more details.
\end{maremarque}

\begin{monlemme} \label{intersectionEmpty} The set of $\xi$-twisted selfdual skew cyclic codes is in bijection with the set of nontwisted selfdual skew cyclic codes and their intersection is empty if $\Frob^s(\xi) \neq \xi$.
\end{monlemme}
\begin{proof}
    We have for any monic skew polynomial $f$ of degree $s$ generating a selfdual code $C_f$ of $\mathbf{E}_{k}^{(l)}$:
    {\noindent $$f \xi  f^{\bullet_{\xi^{-1} }} \xi^{-1}=\frac{f f^\bullet}{X^s(X^r-1)}X^s(X^r-1)= f(0) X^s(X^r-1) $$where $f(0)$ denotes the constant term in $f$. As we assume} $\xi$ to be $\sigma_l$-invariant, by Hilbert-90, we can solve the equation $\gamma \sigma_l(\gamma) = \xi$ for $\gamma$ in ${\mathbf{K}_l}^{\Frob^s}$.
    Noting then $ g = \sigma_l(\gamma) f \gamma^{-1}$, we get a bijection $f \mapsto g$ between nontwisted and $\xi$-twisted selfdual skew cyclic codes:
    {\noindent $$g g^{\bullet_{\xi}}
     =  \sigma_l(\gamma) f f^\bullet \frac{1}{\sigma_l(\gamma)}= f(0) X^s (X^r-1) $$
    Moreover}, if we assume $\Frob^s(\xi)\neq \xi$ and $ff^\bullet=ff^{\bullet_{\xi}}=f\xi^{-1}f^{\bullet}\xi=0$ in $\mathbf{E}_{k}^{(l)}$, then by evaluating lifts at $0$, we get $f(0)=f(0)\frac{\Frob^s(\xi)}{\xi} $, and so $f(0)=0$ and thus $f=0$, which contradicts the hypothesis.
\end{proof}

Algorithm \ref{fullinsep} is an iterator that enumerates selfdual skew cyclic codes sitting in $\mathbf E_k$. It is exhaustive, in the sense that it lists every selfdual code at least once, but it is slightly redundant. 

\begin{algorithm}
    \caption{Enumeration of purely inseparable selfdual skew cyclic codes}\label{fullinsep} 
    \begin{algorithmic}[1]
        \State $\mathcal{C} \gets $ array of length $p^m$ of maps of iterators on all twisted codes indexed by all possible twists $\xi X^{s(\frac{\deg(f)}{s}\%2)}$ where $\xi$ can be choosen among all representatives of $\mathbb{P}_{\mathbf{F}_l}^r$ in $\mathbf{K}_l$  if $\frac{\deg(f)}{s}$ is even and otherwise among all representatives of $\mathbb{P}_{\mathbf{F}_l}^r$ in $\mathbf{K}_l$ that are antisymmetric relatively to $\Frob^s$.
        \Procedure{RunThroughRemainingCodes($f$)}{}
        \State $i \gets \frac{\deg(f)}{s}$
        \State \IfCustom{$i = p^m$}{
            \State \yield $f$
        }
        \State \ElseCustom{}{
            \State \WhileCustom{$f_i \gets \nextc (\mathcal{C}[i][f(0) X^{s(i\%2)}])$}{
                \State RunThroughRemainingCodes($f_i f$)
            }
        }
        \EndProcedure
        \State RunThroughRemainingCodes(1)
    \end{algorithmic}
    
\end{algorithm}
\begin{montheorem}Algorithm~\ref{fullinsep} is correct and exhaustive \label{correctness}
\end{montheorem}
\begin{proof}
    In order to enumerate all inseparable selfdual skew cyclic codes, at the cost of some redundancy, we can assume without loss of generality (See the last part of the proof, hereunder) that the general solution is a product of twisted selfdual skew cyclic codes $f_1 \ldots f_n$, where the $f_i$ are left monic.
    We start by solving the equation $f_n f_n^\bullet \equiv 0 \mod (X^r-1)$. This has been done in the preceding section.
    Now we obtain a scalar $\kappa_n=\frac{f_n f_n^\bullet}{X^s(X^r-1)}$ which is equal to $f_n(0)$. Let ${\bullet}_{\kappa_n}$ be defined by $f_i^{\bullet_{\kappa_n}}=\sigma_l(\kappa_n)f_i^{{\bullet}}{\kappa_n}^{-1}$.
    The equation becomes $f_{n-1} X^s f_{n-1}^{\bullet_{\kappa_n}} = 0\: (X^r-1)$.
    Solving it, we now obtain a scalar $\kappa_{n-1}=\frac{f_{n-1} X^s f_{n-1}^{\bullet_{\kappa_n}}}{X^{r+s}(X^r-1)}$.
    At the next step, the monomials $X^s$ cancel, and we are back in the Hermitian case: $f_{n-2} f_{n-2}^{\bullet_{\kappa_{n-1}\kappa_n}}= 0 \:(X^r-1)$. And so on so forth, getting alternatively a skew Hermitian (resp. skew Euclidean) and a Hermitian (resp. Euclidean) bilinear form.
    We have to check that the $\kappa_i$ satisfy the required symmetry for the selfdual skew cyclic codes to exist.
    A monic polynomial $f$ satisfying the product criterion: $ff^{\bullet_{\kappa X^{t}}}=0$ in $\mathbf{E}_{k,l}^{({\kappa X^t})}$ has a constant term $f(0)$ satisfying:
    \begin{align*}
    (X^s + \cdots + f(0))\kappa X^t (X^r f(0) +\cdots+ X^s) 
      & = \Frob^s(\kappa)\Frob^{s+t}(f(0))  X^s X^t X^r   +  f(0) \kappa X^t X^s \\
      & \,\,\propto\, X^{r+s+t}-X^{s+t}.
    \end{align*}
    Thus we have:
    \begin{align}
            \Frob^s(f(0))=-f(0)  & \quad  \text{for} \: \Frob^s(\kappa)=\kappa\:,\: t=0    \label{eq:1} \\ 
            \Frob^s(\kappa)=-\kappa  &  \quad \text{for} \:  t=s \, \text{(symplectic case)} \label{eq:2} \\
            -f(0) \kappa= \Frob^s(\kappa f(0)) &  \quad \text{for} \: t=0 \label{eq:3}
    \end{align}
    If we start with $\kappa=1$, we get the symplectic case from~\eqref{eq:1} and~\eqref{eq:2} with $\kappa$ satisfying  $\Frob^s(\kappa)=-\kappa$, at the next step. We have then an orthogonal case, then again alternatively a symplectic case with $\kappa$ satisfying  $\Frob^s(\kappa)=-\kappa$ from~\eqref{eq:3}, \emph{etc}.

    We now prove that the algorithm is exhaustive.
    We observe that the projection
    $$\begin{array}{rcl}
        \mathbf{E}_{k}^{(l)}= {\mathbf{K}_l}[X^{\pm 1};\Frobl]/(X^r-1)^{p^m}  &\longrightarrow & {\mathbf{K}_l}[X^{\pm 1};\Frobl]/(X^r-1)=\mathbf{E}_{1}^{(l)} \\
        f &\mapsto & \bar{f}
    \end{array}$$
    preserves the selforthogonality property.
    Noting $f_{p^m}:=\tilde{\bar{f}}$ the unique lift of $\bar{f}$ on the basis $(X^i)_{0\leq i < r}$, we have a factorization $f_{p^m}=r_{p^m}g_{p^m}$ for a skew polynomial $r_{p^m}$ of degree striclty less than $s$ and a selfdual skew cyclic code $g_{p^m}$ in $E_1$.
    Indeed any selforthogonal subspace of $\mathbf{E}_{1}^{(l)}$ of dimension strictly less than $r$, corresponding to a monic skew polynomial $f$ can be extended, by Witt's decomposition, to a maximal isotropic space corresponding to a selfdual monic skew polynomial $g$ of degree $s$. 
    Now this vector space inclusion corresponds by duality to a factorization $\bar{f}=r g$ for a skew polynomial $r$ of degree striclty less than $s$.
    Expressing $(X^r-1)$ as a product of the selfdual codes $g_{p^m}$, $\frac{g_{p^m}^*g_{p^m}}{g_{p^m}(0)X^s}$ we get that any selforthogonal skew cyclic code $f \in \mathbf{E}_{k}^{(l)}$ can be written in the form 
    $$f=h(X^r-1)+f_{p^m}=\left(h\frac{g_{p^m}^*}{g_{p^m}(0)X^s}+r_{p^m}\right)g_{p^m}$$where $\deg{h}=(p^m-2)s$ and $\frac{1}{X^s g_{p^m}(0)}g_{p^m}^* \in {\mathbf{K}_l}[X;\Frob]$ is of degree $s$.
    Let us note $f'=g\frac{1}{X^s g_{p^m}(0)}g_{p^m}^*+r_{p^m}$. We have $\deg{f'}=(p^m-1)s$ and $f' g_{p^m}(f' g_{p^m})^*=f' g(0)X^s f'^* (X^r-1) \equiv 0 \mod (X^r-1)^{p^m}$ and hence $f'f'^{\bullet_{X^{s}g_{p^m}(0) }} \equiv 0 \mod (X^r-1)^{p^m-1}$.
    With the same reasoning, replacing $f$ by $f'$ and the adjunction $*$ by $\bullet_{X^{s}g_{p^m}(0)}$, we get yet another twisted separable selfdual skew cyclic code $f_{p^m-1}$ and another skew polynomial $f''$ such that  $f'' f''^{\bullet_{g_{p^m-1}(0)g_{p^m}(0) }} \equiv 0 \mod (X^r-1)^{p^m-2}$.
    In turn replacing $f'$ by $f''$ and the adjunction $\bullet_{X^{s}g_{p^m}(0) }$ by $\bullet_{g_{p^m-1}(0)g_{p^m}(0) }$, we get yet another twisted separable selfdual skew cyclic code $g_{p^m-2}$ and another skew polynomial $f''$ such that $f''  f''^{\bullet_{ X^s g_{p^m-2}(0)g_{p^m-1}(0)g_{p^m}(0) }} \equiv 0 \mod (X^r-1)^{p^m-3}$.
    Per induction we thus a factorization $g_0 g_1 \cdots g_{p^m-1} g_{p^m}$ of $f$ into twisted separable selfdual skew cyclic codes as claimed.
\end{proof}
\begin{maremarque}We notice the reason for the redundancy in the enumeration algorithm from the above proof of the exhaustivity. Indeed the many different factorizations $f_{i}=r_{i}g_{i}$ for selforthogonal $f_i$ at each step lead to as many redundant enumerations of the same inseparable selfdual skew cyclic code $f$. 
\end{maremarque}

\subsection{SageMath enumeration of inseparable selfdual skew cyclic codes}

For $F=GF(3)$, $K=GF(3^6)$ and $k=3$, the upper bound on the number of generated inseparable selfdual skew cyclic codes is numerically equal to $80 \times 1120 \times 80$, where $80$ is the number of orthogonal isotropic spaces and $1120$ the number of symplectic isotropic spaces (see Remark~\ref{rem:segre2}).
A SageMath enumeration based on this algorithm provides a number $n$ of maximal isotropic codes equal to $n = 2360960$.
We have not many redundancies since $80 \times 1120 \times 80 \approx 3 \times 2360960$.
For the purpose of this heavy computation we implemented the PARI/GP optimization for finite field extensions in a dedicated branch of our code, which is only valid for prime base fields.
The computation takes place in less than 10 minutes on the aforementioned computer.


\bibliographystyle{alpha}
\bibliography{selfOrthogonalCodes}

\end{document}